\documentclass[12pt,a4paper]{article}

\usepackage[margin=20mm]{geometry}
\usepackage[T1]{fontenc}
\usepackage[latin2]{inputenc}
\usepackage{ae}
\usepackage{graphicx}
\usepackage[american]{babel}
\usepackage{amsfonts}

\usepackage{amsthm}
\usepackage{amsmath}
\usepackage{amssymb}
\usepackage{wasysym}
\usepackage{paralist}
\usepackage{multirow}
\usepackage{cite}
\usepackage{pifont}
\usepackage{imakeidx}
\usepackage{arydshln}

\makeindex
\newtheorem{Theorem}{Theorem}[section]
\newtheorem{Def}[Theorem]{Definition}
\newtheorem{Lemma}[Theorem]{Lemma}
\newtheorem{Corollary}[Theorem]{Corollary}
\newtheorem{Proposition}[Theorem]{Proposition}
\newtheorem{Question}[Theorem]{Question}
\newtheorem{Conjecture}[Theorem]{Conjecture}
\newtheorem{Problem}[Theorem]{Problem}
\newtheorem{Note}[Theorem]{Note}

\renewcommand{\d}{\ {\rm d}}
\newcommand{\dist}{{\rm dist}}

\newcommand{\eps}{{\varepsilon}}

\newcommand{\N}{\mathbb N}

\newcommand{\R}{\mathbb R}
\newcommand{\Z}{\mathbb Z}
\renewcommand{\:}{\colon}

\newcommand{\bigwhere}{\ \big|\ }
\newcommand{\Bigwhere}{\ \Big|\ }

\usepackage{color}
\definecolor{red}{rgb}{0.8,0,0}
\definecolor{green}{rgb}{0,0.6,0}

\title{Limit theory of discrete mathematics problems}
\author{Endre Csóka\thanks{Supported by Marie Sk\l odowska-Curie grant 750857, ERC grants 306493
and 648017 and MTA Rényi ``Lendület'' Groups and Graphs Research Group.}\\{\normalsize Alfréd Rényi Institute of Mathematics, Budapest, Hungary}}
\date{}

\begin{document}
\maketitle
\begin{abstract}

We show a general problem-solving tool called limit theory. This is an advanced version of asymptotic analysis of discrete problems when some finite parameter tends to infinity. We will apply it on three closely related problems.

Alpern's Caching Game (for 2 nuts) is defined as follows. The hider caches 2 nuts into one or two of $n$ potential holes by digging at most 1 depth in total. The goal of the searcher is to find both nuts in a limited time $h$, otherwise the hider wins.
We will show that if $h$ and $n/h$ are large enough, then very counterintuitively, any optimal hiding strategy should dig less than 1 in total, with positive probability. We will prove it by defining and analyzing a limit problem. Then we will partially solve the entire problem.
We will also have significant progress with two other problems: the Manickam--Miklós--Singhi Conjecture and the Kikuta--Ruckle Conjecture.

\end{abstract}

\section{Introduction}

For finite problems, it is a very common technique to solve infinite or continuous versions of it, and its solution may be useful for the original problem. Another very common technique is that we find the asymptotics of the solution when, e.g., a parameter tends to infinity. Limit theory is a combination of these two techniques. Namely, we define a limit problem about which we can prove that its solution is the limit of the solutions of a sequence of finite problems.
Finding a good limit problem can be difficult but very useful.

Limit theory techniques were already used in different areas. Statistical physics is essentially just limit theory in physics. In mathematics, probably Fürstenberg (1977, \cite{furstenberg1977ergodic}) was the first to use limit theory techniques, for reproving Szemerédi's Theorem. Lovász and Szegedy (2006, \cite{lovasz2006limits}) used this technique by introducing limit graphs called graphons. This also motivated the limit theory of many different discrete structures. However, in all these cases, limit theory was used only for special purposes. In this paper, we will show through several examples that this is indeed a general problem-solving tool.

Our most convincing but most difficult application is about Alpern's Caching Game (Section~\ref{Alpernsec}). We will find and analyze multiple limit problems, and we will get to some very interesting and highly counterintuitive results.
Then we will extend the Manickam--Miklós--Singhi Conjecture (MMS, Section~\ref{MMSsec}) using a simple but nontrivial limit problem. Finally, we will consider the Kikuta--Ruckle Problem (KR, Section~\ref{KRsec}), which is a generalization of the MMS Problem. We will understand and specify the KR Conjecture in a highly nontrivial way, using some limit problems.
About Alpern's Caching Game and the KR Problem, the solutions depending on the parameters looked chaotic, but our techniques will reveal the structures in them.

\section{Alpern's Caching Game} \label{Alpernsec}


\begin{Def}
Alpern's Caching Game $G(k, j, n, h)$ is defined as the following 2-player game between the hider and the searcher.
There are $n \in N$ (potential) holes and $k \in N$ nuts.
The hider places (caches) each of $k$ nuts into one of the holes in a positive depth, so that the total digging time (the sum of the depths of the deepest nut in each non-empty hole) must be at most 1.
The searcher cannot observe anything about the placement, but he can dig the hole at most depth $h \in \R^+$ in total.
A nut is found if the searcher dug at least as much in that hole as the depth of the nut.
The searcher can choose an adaptive digging strategy, continuously observing what and where he already found during the digging.\footnote{
Strategy means mixed (or randomized) strategy. If we use the discrete sigma-algebra on the set of strategies, then there is nothing to add to the definition, but we cannot use continuous distributions about the hiding depths and we will have only approximately optimal strategies. Or we can use the Lebesgue sigma-algebra on the set of hiding strategies, but in this case, we need some measurability criteria for the searching strategy in order to define winning probabilities. These are irrelevant issues about our results, and we will omit these technical details throughout the section.
}
The searcher wins if he finds at least $j \in N$ out of the $k$ nuts. Otherwise the hider wins.
\end{Def}


\medskip

This is a 2-player 0-sum game, therefore, there is a \textbf{value of the game} $v = v(k, j, n, h)$ with the following properties.
The searcher has a strategy which wins with probability at least $v$ against any pure (= deterministic) hiding strategy called \emph{placement}, and the hider has a strategy which wins with probability at least $1-v$ against any pure searching strategy. These are called optimal strategies.

\subsection{Previous results}

This problem was introduced by Alpern, Fokkink, Lidbetter and Clayton in \cite{alpern2012search}, a summary about the results and related questions of Alpern's Caching Game can be found in the survey book by Alpern, Fokkink, Leszek, Lindelauf \cite{alpern2013search}. More recent results are presented in \cite{csoka2016solution}.

The problem is solved for $k = j = 2$, $n \le 4$.\cite{alpern2012search, csoka2016solution} The solution for $k = j = 2$, $n = 4$ is the stepfunction shown in Figure~\ref{Alpern-4}. We show the nature of these optimal strategies by the following few cases.

\begin{figure}[]
\begin{center}
\includegraphics[width=\textwidth, trim=50mm 10mm 50mm 20mm]{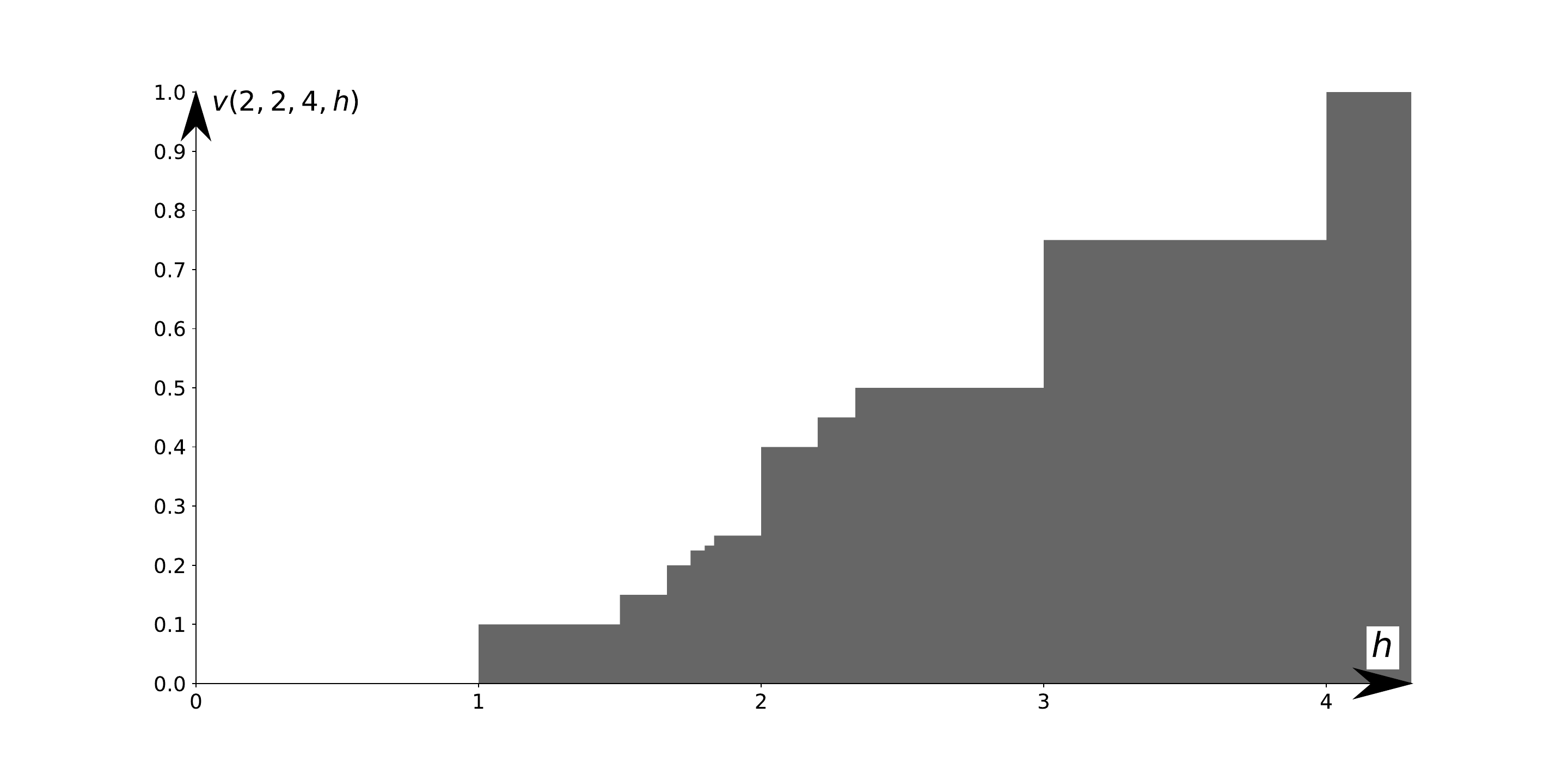}
\noindent
{\scriptsize
\begin{tabular}{|c||c|c|c|c|c|c|c|c|c|c|c|c|}
  \hline
  \phantom{\Big( } $h$ & $[0, 1)$ & $\big[1, \frac{3}{2}\big)$ & $\big[\frac{3}{2}, \frac{5}{3}\big)$ & $\big[\frac{5}{3}, \frac{7}{4}\big)$ & $\big[\frac{7}{4}, \frac{9}{5}\big)$ & $\big[\frac{9}{5}, \frac{11}{6}\big)$ & $ \big[\frac{11}{6}, 2\big)$ & $\big[2, \frac{11}{5}\big)$ & $\big[\frac{11}{5}, \frac{7}{3}\big)$ & $\big[\frac{7}{3}, 3\big)$ & $[3, 4)$ & 4 \\
  \hline
  \phantom{\Big( }$v(2, 2, 4, h)$ & $0$ & $\frac{1}{10}$ & $\frac{3}{20}$ & $\frac{1}{5}$ & $\frac{9}{40}$ & $\frac{7}{30}$ & $\frac{1}{4}$ & $\frac{2}{5}$ & $\frac{9}{20}$ & $\frac{1}{2}$ & $\frac{3}{4}$ & 1 \\
  \hline
\end{tabular}
}

\medskip
\caption{The solution value of $G(2,2,4,h)$.} \label{Alpern-4}
\end{center}
\end{figure}

\begin{itemize}
  \item If $h \ge 1$, then the searcher makes a uniform random guess for the two holes of the two objects (10 options including the cases that the two holes are the same), and he digs in those holes until he find them. If the guess was right, then he finds both. Therefore, he wins with probability at least $\frac{1}{10}$.

      On the other hand, the hider can choose the two holes uniformly at random out of the 10 options, and hide them at depths $\frac{1}{2}, \frac{1}{2}$ for two different holes, or $\frac{1}{2}, 1$ if he chose the same hole. If $h < \frac{3}{2}$, then the searcher cannot find both in more than 1 out of the 10 placements.
  \item If $h \ge \frac{3}{2}$, then the searcher does the same as for $h \ge 1$ except that if the first guess was right but the second guess was wrong, then he makes another guess for the second hole, and if the nut is there at depth at most $\frac{1}{2}$, then he finds it. If the searcher makes the choices with the right probabilities, then the searcher wins with probability at least $\frac{3}{20}$ for each placement.

      On the other hand, if the hider chooses depths $\frac{1}{3}, \frac{2}{3}$ for different holes, or $\frac{1}{3}, 1$ or $\frac{2}{3}, 1$ for same holes, then the searcher cannot find them in more than 3 cases if $h < \frac{5}{3}$. Therefore, if the hider chooses the placement uniformly at random out of the 20 possibilities, then the searcher wins with probability at most $\frac{3}{20}$.
  \item If $h < \frac{7}{3}$, then the optimal hiding strategy has the same structure as for $h \in \big[\frac{3}{2}, \frac{5}{3}\big)$, except that instead of depth $\frac{1}{3}$ (and $1 - \frac{1}{3} = \frac{2}{3}$), we are using depth $\frac{1}{6}$ or $\frac{1}{5}$ or $\frac{1}{4}$ or $\frac{2}{5}$ or $\frac{1}{2}$, and the hider uses a convex combination of these hiding strategies. If $h \ge \frac{7}{3}$, then the hider simply puts both nuts to the same hole. The optimal searching strategies are much more difficult and very irregular, those would require many pages to describe.
\end{itemize}

These results suggested that the solution even for $k = j = 2$ nuts (and for arbitrary $n$ and $h$) are chaotic and almost hopeless to characterize. Therefore, the researchers of the question claimed that they gave up trying to solve this problem. (Personal communication.)

\subsection{The breakthrough using limit theory}

Using limit theory, we will show that the solution is difficult but not that chaotic, and very different from what the earlier results suggested. First, we show a surprising property of the solution, and then we will partially solve the problem for $k = j = 2$. Some of the results will apply for $k = j > 2$ and $k = j + 1 > 2$.

\begin{Def}
In Alpern's Caching Game, we say that a placement of the nuts is \textbf{extremal} if
this requires total digging depth $1$ (the sum of the depth of the deepest nut in the non-empty holes is exactly 1). A (mixed) hiding strategy is called extremal if it is supported on extremal placements. The extremal version of the game $XG$ means that the hider must use an extremal strategy.
\end{Def}

Let $v_X$ always denote the same as $v$ with the extremal version of the game.

\begin{Question} \label{QX}
Does the hider always have an extremal strategy which is optimal? \\ Or (equivalently) does $v(k, j, n, h) = v_X(k, j, n, h)$ always hold?
\end{Question}

The answer was believed to be clearly positive, some of the authors of the problem did not even realize that they did not have a proper proof of it (according to private communications). Moreover, there was a conjecture presented in \cite{alpern2013search} about a difficult recursive property of the optimal strategies of the hider, which was in accordance to the solved cases, but which implied the positive answer to Question~\ref{QX}.

In the beginning, the author of this paper was almost sure about the positive answer, too. But limit theory analysis pointed out the opposite. On the top of it, further analysis showed that for a large class of parameters, the optimal hiding strategy uses non-extremal placements with probability 1.

The proof of our final results are presented in Subsection~\ref{Alpern-original-sol}. However, the primary goal of this section is not showing the final results and the proofs like pulling a rabbit out of a hat, but we want to show the technique of how we found these results. The same technique will be used for the other two problems in this paper.

\medskip

Define the limit of the game when the number of nuts to hide $k$ and to find $j$ are fixed, but the number of holes and the digging time $n, h \rightarrow \infty$, with an asymptotic ratio $n / h \rightarrow \lambda$.

\begin{Def}
The \textbf{limit game} $LG(k, j, \lambda)$ and its extremal version $XLG(k, j, \lambda)$ are defined as follows. The hider chooses a partitioning $a_1, a_2, ..., a_{k'} \in \Z^+$ where $\sum a_i = k$, and he chooses values (depths) $y_1, y_2, ... y_{k'} \in [0, 1]$ with $\sum y_i \le 1$ in $LG$, and $\sum y_i = 1$ in $XLG$. Then for $k'$ independent uniform random numbers $x_1, x_2, ... x_{k'} \in [0, \lambda]$, $a_i$ number of nuts are placed at $(x_i, y_i)$. The searcher observes nothing. Now the searcher should define a function $f_t(x): [0, 1] \times [0, \lambda] \rightarrow [0, 1]$ which is monotone increasing in both parameters and $\int f_1(x) \d x \le 1$. Then we evaluate $f$ meaning that the searcher gets to know the smallest $t^*$ so that $f_{t^*}(x_i) \ge y_i$ for some $i$. If there is no such nut position even for $f_1$, then the game ends. Otherwise the nuts at $(x_i, y_i)$ are found by the searcher, he gets to know the position and the number of them, and we remove these nuts. Then the searcher can change his function in the parameter interval $t \in (t^*, 1]$, and we re-evaluate $f$. The searcher wins if he finds at least $j$ nuts in total.
\end{Def}

\begin{Note}
We can get an equivalent problem by assuming that $\int f_t(x) \d x = t$ for all $t \in [0, 1]$. This will be convenient to assume when we are showing upper bounds on the value of the game. Also, we can omit the condition that $f$ is monotone increasing in the second coordinate, which will be useful for lower bounds.
\end{Note}

The following theorem shows the reason why we call it a limit game.
Denote the value of the limit game by $v(k, j, \lambda)$.

\begin{Theorem} \label{limitthm}
For any parameters $k, j, n, h \in \Z^+$,
\begin{equation*}
    v(k, j, n, h) : v\Big(k, j, \frac{n}{h}\Big) \in \Big[\Big(\frac{h-j}{h}\Big)^j, 1\Big].
\end{equation*}
Therefore, if $n_i / h_i \rightarrow \lambda$ and $n \rightarrow \infty$, then $v(k, j, n_i, h_i) \rightarrow v(k, j, \lambda)$.
The same applies for the extremal versions.
\end{Theorem}

This theorem will be used when we will disprove Question~\ref{QX}. But we will not need this for the final results.
Therefore, we present just a sketch of proof of the theorem.

\begin{proof}[Sketch of proof]
Notice first that both players can choose a uniform random permutation of the holes and apply his strategy on this permutation. If either of them does so, then whether the other player does it makes no difference in the expected result. Therefore, if we add to the rules that either or both players must use this randomization, then it does not change the value of the game, as both players can secure himself this expected score.


About the limit games $LG$ and $XLG$, notice that if two holes are dug to the same depth, and nothing was found in them so far, then it does not matter which one the searcher continues digging. Therefore, we can assume that according to the random ordering of the holes, their depths remain monotone decreasing during the search, excluding holes in which we have already found a nut.

Now let us see why do the values of the discrete problems converge to the values of the limit problems.

On one hand, any strategy of the searcher in the discrete game $G$ or $XG$ can be applied in the limit game by choosing $f_t$ in the interval $\big[\frac{i-1}{n},\frac{i}{n}\big)$ as the depth of the $i$th deepest hole after a total amount of digging $t$. This way the searcher can get at least the same score as in the discrete game.

On the other hand, a strategy of the searcher in $LG$ (or $XLG$) can be applied in $G$ (or $XG$) as follows. The searcher chooses a random ordering of the holes. Then he digs so as to have depth $f_t\big(\frac{i-1}{h-j}\big)$ in the $i$th hole, except that if a nut is found in a hole, then he digs that hole until depth $1$. He does it for all $t \in [0, 1]$, in increasing order. This way, the searcher can get at least $\big( \frac{h-j}{h} \big)^j$ times the score of the limit game $LG$ (or $XLG$).
\end{proof}

\begin{Def}
The \textbf{double limit game} $DLG(k, j)$ and its extremal version $XDLG(k, j)$ are defined as the limit game with $\lambda \rightarrow \infty$, as follows.
The hider chooses $k$ values $y_1, y_2, ... y_k \in [0, 1]$ where $\sum y_i \le 1$ in $DLG$, and $\sum y_i = 1$ in $XDLG$.
At the same time, the searcher defines a pure strategy of the limit game with $\lambda = \infty$.
Then for each subset $Q = \{ q_1, q_2, ..., q_j \} \subset \{1, 2, ..., k\}$ with a vector of positive real numbers $x_{q_1}, x_{q_2}, ... x_{q_j} \in \R^+$, the nuts are placed at $(x_{q_i}, y_{q_i})$. The score of the searcher is the $j$-dimensional measure of the vectors $x_{q_1}, x_{q_2}, ... x_{q_j}$ for which he wins by his strategy, summing up for all different $Q$. This is what the searcher aims to maximize and the hider aims to minimize, in expectation.
\end{Def}

Denote the value of $DLG(k, j)$ by $v(k, j)$.

\begin{Theorem} \label{2limitthm}
Fix $k \ge j \ge 2$, and consider a sequence of pairs $(n_i, h_i)$ so that $h_i \rightarrow \infty$ and $\frac{n_i}{h_i} \rightarrow \infty$. Then
\begin{equation*}
  \Big(\frac{n_i}{h_i}\Big)^j \cdot v(k, j, n_i, h_i) \rightarrow v(k, j).
\end{equation*}
The same applies for the extremal versions.
\end{Theorem}

\begin{proof}[Sketch of proof]
The strategy of the searcher in $LG$ can be applied in $DLG$. This shows one direction.

The other direction is a bit more technical. The first observation is that in $G$, if the hider puts more nuts in the same hole, and the searcher digs $\lfloor h \rfloor$ holes until depth 1, then this already provides him a score $\omega\big(\big(\frac{n_i}{h_i}\big)^{-j}\big)$. Therefore, in the optimal hiding strategy of the hider, the probability of such a placement tends to 0. Thus, the limit of the values does not change if we forbid such a placement in $G$.

The next observation is that the probability of finding more than $j$ nuts is $o\big(\big(\frac{n_i}{h_i}\big)^{-j}\big)$. Therefore, the probability of finding $j$ nuts is essentially the same as the expected number of $j$-element subsets of the nuts which would be found by the searcher if the other nuts had not been cached.

The optimal strategy of the searcher in $DLG$ can be applied in $LG$ with a large parameter $\lambda$, simply by restricting $f$ to $[0, 1] \times [0, \lambda]$. In $DLG$, if $x > \lambda$, then (by monotonicity) $f_1(x) < \frac{1}{\lambda}$ throughout the game. Therefore, this restricted strategy provides the same score unless if the depth of a nut is at most $\frac{1}{\lambda}$.

The following searching strategy is very efficient if the depth of a nut is at most $\frac{1}{\lambda}$. First, the searcher chooses $f_t(x) = \frac{1}{\lambda}$ if $x < \frac{t}{\lambda}$, and 0 otherwise. Then, after finding the first nut at $(x_1, y_1)$, then he chooses $f_1(x) = \frac{1}{\lambda}$ if $x < x_1$, $f_1(x) = 1$ if $x \in [x_1,\ x_1 + 1 - \lambda x_1]$ and 0 otherwise. If we use this strategy with probability $O\big(\frac{1}{\lambda}\big)$ and the strategy $f$ restricted to $[0, 1] \times [0, \lambda]$ otherwise, then for all possible hiding strategy, this mixed searching strategy in $LG$ will be (at least) almost as good as the original strategy in $DLG$.

The same argument works for the extremal games, as well.
\end{proof}


Consider any optimal strategy of the hider in $XDLG(2, 2)$. This can be identified with the probability measure $\mu$ of the depth of a random nut.

\begin{Lemma} \label{round2opt}
In $XDLG(2, 2)$, if the searcher with any optimal strategy finds a nut at depth $y \in \operatorname{supp}(\mu)$,
then he almost always changes $f$ so as to maximize the size of the interval $f_1^{-1}(1-y) = \big\{x \in [0, \infty) : f_1(x) = 1 - y\big\}$.
\\ More precisely, if the hider chooses a placement randomly from $\mu$, and the searcher plays optimally, then finding a nut and changing $f$ not in the suggested way happens with probability 0.
\end{Lemma}

\begin{proof}[Sketch of proof]
  Otherwise the searcher could improve his score against an optimal hiding strategy.
\end{proof}

Two pure searching strategies are called \emph{equivalent} if they get the same score against any hiding strategies including the non-extremal ones. Two (mixed) searching strategies are equivalent if there exists a measure preserving bijection between the two distributions of pure strategies such that the corresponding pure strategies are equivalent except for a 0-measure set.

\begin{Lemma} \label{round1opt}
For any optimal strategy of the searcher in $XDLG(2, 2)$, there is an equivalent one which starts with a function $f$ satisfying that for all $(t, x) \in [0, 1] \times [0, \infty)$,
\begin{equation} \label{round1opt-eq}
f_t(x) = 0 \hskip 5mm \text{or} \hskip 5mm \forall \eps > 0\: \mu\big[f_t(x) - \eps,\ f_t(x)\big] > 0.
\end{equation}
\noindent (This is a little more restrictive than $f_t(x) \in \{0\} \cup \operatorname{supp}(\mu)$.)
\end{Lemma}

\begin{proof}[Sketch of proof]
Assume that $\int f_t(x) \d x = t$. Let
\begin{equation*}
q_t(x) = \sup\Big\{y \le f_t(x) \Bigwhere (y = 0)  \text{ or } \big(\forall \eps > 0: \mu[y - \eps,\ y] > 0\big) \Big\}.
\end{equation*}
Clearly, $f$ and $q$ get the same score against the hiding strategy $\mu$.
Now, after some case analysis, we can get to the following conclusion. Either $q$ is an equivalent searching strategy to $f$, or we can find an $\hat{f} \ge q$ which provides higher score than $f$ against $\mu$.
\end{proof}


\begin{Theorem} \label{limitcounter}
$v_X(2, 2) < v(2, 2)$.
\end{Theorem}


\begin{Corollary}
Theorems \ref{2limitthm} and \ref{limitcounter} imply the existence of infinitely many counterexamples for Question~\ref{QX}. \qed
\end{Corollary}

\begin{proof}[Sketch of proof of Theorem~\ref{limitcounter}]

Assume by contradiction that $v_X(2, 2) = v(2, 2) = v$. 
Consider an optimal strategy $S'$ of the searcher in $DLG$. This provides the expected score at least $v$ against any strategy of the hider. Therefore, $S'$ provides an expected score at least $v$ against all extremal hiding strategies, therefore, $S'$ is an optimal searching strategy in $XDLG$, as well. Lemma~\ref{round1opt} says that there exists an $S$ equivalent to $S'$ satisfying \eqref{round1opt-eq}. Consequently, $S$ is an optimal searching strategy in $DLG$ which satisfies \eqref{round1opt-eq}. This can be represented by the first-round searching function $f^S$.

$\mu$ and $S$ are optimal hiding and searching strategies in both $DLG$ and $XDLG$, because they provide an expected score at least and at most $v$, respectively, against any searching strategy.
According to Lemma~\ref{round2opt}, the strategy of the searcher is represented by the first function $f$ he chooses.

\textbf{Case 1.} $\operatorname{supp}(\mu) = [0, 1]$. (We conjecture this to be true.)

In $XDLG$, a randomization between the following two pure strategies of the searcher provides him a score at least $\sqrt{2} + 1$.
\begin{itemize}
\item $f_t(x) = 1$ if $x < t$, otherwise 0. \hskip 3mm -- \hskip 3mm This scores $\frac{1}{2 y_1 y_2}$.
\item $f_t(x) = \frac{1}{2}$ if $x < 2t$, otherwise 0. \hskip 3mm -- \hskip 3mm This scores $\frac{1}{y_2} + \frac{1}{2y_2^2}$, where $y_2 \ge y_1$.
\end{itemize}
Therefore,
\begin{equation} \label{v>1+sqrt2}
v \ge \sqrt{2} + 1.
\end{equation}

Consider an arbitrary strategy of the searcher in $XDLG$.
Let $t^*$ denote the time point when $\big(\frac{4}{3}, \frac{1}{4}\big)$ is dug, and let $s_1$ and $s_2$ denote the length of holes which had been dug to the depth at least $\frac{1}{4}$ before or by $t^*$, respectively. Formally,
\begin{equation*}
t^* = \sup_{t \in [0, 1]} \bigg\{ f_t \Big(\frac{4}{3}\Big) < \frac{1}{4} \bigg\}, \hskip 8mm
s_1 = \inf_{x \in \R} \bigg\{ \sup_{t \in [0, 1]} \Big\{ f_t (x) < \frac{1}{4} \Big\} = t^* \bigg\}, \hskip 8mm
s_2 = \sup_{x \in \R} \Big\{ f_{t^*} (x) \ge \frac{1}{4} \Big\}.
\end{equation*}

Now consider the score it provides against the hiding strategies $y_1 = \frac{1}{4} - \eps$ and $y_2 = \frac{1}{4}$ when $\eps \rightarrow 0+$. In the limit, the searcher can win only if one of the followings are satisfied.
\begin{itemize}
  \item $\max(x_1, x_2) \le s_1$ and $x_1, x_2 \le \frac{4}{3}$
  \item $x_1 \in [s_1, s_2]$ and $x_2 \in \big[s_1,\ 2 - \frac{s_2}{2}\big]$
\end{itemize}
The total area of the set of these pairs $(x_1, x_2) \in [0, \infty)^2$ is 
\begin{equation*}
2 \cdot \frac{4}{3} s_1 - s_1^2 + \max\Big(0,\, (s_2 - s_1) \big(2 - \frac{s_2}{2} - s_1\big)\Big) \le 2
\end{equation*}
with equation if $s_1 = 0$ and $s_2 = 2$. This contradicts with \eqref{v>1+sqrt2}.


\medskip

\textbf{Case 2.} $\operatorname{supp}(\mu) \ne [0, 1]$. We only consider the case when there exist $0 < a < b < 1$ satisfying $\mu(a, b) = \mu(1-b,\ 1-a) = 0$, but $\mu(a), \mu(b) > 0$. The proofs of the other cases are essentially the same.

Compare the score of the two hiding strategies $H_1 = (a,\ 1 - a)$ and $H_2 = \big(a,\ 1 - \frac{a+b}{2}\big)$ against $S$. Compare them when the same pairs of holes $(x_1, x_2)$ are chosen. Lemma~\ref{round1opt} implies that the first nut is found at the same time point and in the same hole in the two cases. If this is the first nut (at $(x_1, a)$), then Lemma~\ref{round2opt} shows that the searcher digs either both or none of the two points $\big(x_2,\ 1 - \frac{a+b}{2}\big)$ and $(x_2,\ 1 - a)$, therefore, $H_1$ and $H_2$ are equally good in this case. If the second nut was found for first, then the other nut is in the same place $(x_1, a)$ in the two cases, and Lemma~\ref{round2opt} shows that $S$ plays optimally against $H_1$ after finding this nut. This implies that $S$ gets at most as much score against $H_2$ than against $H_1$.

The optimality of $S$ and $\mu$ with the fact that $\mu(1-a) = \mu(a) > 0$ imply that $S$ gets the score $v$ against $H_1$. But $S$ gets at least $v$ against all hiding strategies. Therefore, $S$ gets the score $v$ against $H_2$, as well. This means that if $S$ finds the second nut, then he plays optimally also against $H_2$, meaning that it no longer digs deeper than $a$.

We can use the same argument with $H^{\prime}_2 = \big(\frac{a+b}{2},\ 1 - a\big)$. Therefore, if the hider chooses depths $\big(\frac{a+b}{2},\ 1 - \frac{a+b}{2}\big)$, then $S$, after finding one nut, completely fails to dig at the right depth for the other nut, and hereby the searcher gets the score 0. This contradicts with the optimality of $S$.
\end{proof}


\subsection{Solution for the double limit game} \label{sol-dlg}

For first, it seemed that the extremal double limit game is easier to solve than the double limit game. The reason of it is that the strategy of the hider is a probability distribution on an interval in $XDLG$, and on a two-dimensional domain in $DLG$. And therefore, the author expected that $XDLG$ is easier to solve than $DLG$. But the truth seems to be the opposite.

The extremal double limit game is not solved yet. If somebody tries to solve it, then the author suggests considering the searcher's function $f_{t - 1}(x) = \chi(x < t) \cdot \big(1 - \frac{1}{t}\big)$ with probability more than $\frac{1}{2}$. The searcher's other pure strategy may start with $f_t(x) = \chi(x < t) g(x)$ for $t \in [0, \eps]$ with a function $g(0) = 1$, $g'(0) \approx -0.1$. The author believes that $v_X(2, 2) \approx 2.8$.

On the other hand, the double limit game has a surprisingly simple solution, as follows.

\begin{Theorem} \label{limitsolution}
If $\lambda \ge k = j$, then $v(k, j, \lambda) = v(k, k, \lambda) = \frac{k!}{\lambda^k}$, and therefore, the value of the double limit game is $v(k, k) = k!$~.
\end{Theorem}

The proof of Theorem~\ref{limitsolution} consists of Propositions~\ref{hiderstrategy} and \ref{searchersrtategy}, showing optimal strategies for the hider and the searcher.

\begin{Proposition} \label{hiderstrategy}
In $LG(k, k, \lambda)$ with $\lambda \ge k$, if the hider chooses a uniform random point $(y_1, y_2, ..., y_k)$ from the simplex $y_i > 0\ (\forall i \in \{1, 2, ..., k\})$, $\sum y_i \le 1$, then the searcher wins with probability at most $\frac{k!}{\lambda^k}$.

\noindent
(Moreover, the searcher wins with exactly this probability provided that he always searches for nuts in places where it is possible to find one (e.g. never in depth $> 1$).)

\end{Proposition}

\begin{proof}
Consider the measure space $T$ of $k$ time points $0 \le t_1 \le t_2 \le ... \le t_k \le 1$.
For each strategy of the searcher, consider the measure space $S$ of all allocations of the nuts for which the searcher would find all nuts.
To each allocation in $S$, we can assign the vector of time points when the searcher finds the nuts. This is an injective mapping from $S$ to $T$, and the inverse of it is measure-preserving. Therefore, the measure of $S$ is at most the measure of $T$.
The allocation of the nuts is a uniform random point $(x_1, y_1, x_2, y_2, ..., x_k, y_k)$ from the set $x_i \in [0, \lambda]$, $y_i > 0$, $\sum y_i \le 1$, but this set is factored by the $k!$ permutations of the $k$ indices. Therefore, the winning probability of the searcher is at most
\begin{equation*}
\frac{\operatorname{Vol}\big((t_1, t_2, ..., t_k) \bigwhere 0 \le t_1 \le t_2 \le ... \le t_k \le 1 \big)}{\frac{1}{k!} \cdot \lambda^k \cdot \operatorname{Vol}\big((y_1, y_2, ..., y_k) \bigwhere (\forall i\: y_i \ge 0), \sum y_i \le 1\big)}
\end{equation*}
\begin{equation} \label{limitproofeq}
= \frac{k! \operatorname{Vol}\big((t_1, t_2, ..., t_k) \bigwhere t_1, t_2 - t_1, t_3 - t_2, ..., t_k - t_{k-1} \ge 0; t_k \le 1 \big)}{\lambda^k \cdot \operatorname{Vol}\big((y_1, y_2, ..., y_k) \bigwhere (\forall i\: y_i \ge 0), \sum y_i \le 1\big)}
= \frac{k!}{\lambda^k}. \qedhere
\end{equation}
\end{proof}

\begin{Proposition} \label{searchersrtategy}
In $LG(k, k, \lambda)$ for $\lambda \ge k$, the searcher can win with probability at least $\frac{k!}{\lambda^k}$ by the following strategy.

He digs parallelly in a unit interval, and if he finds a nut, then he goes to the next interval. Formally, if he found so far $q$ nuts at the points in time $t_1, t_2, ..., t_q$, then with $t_0 = 0$, he chooses the function
\begin{equation*}
  f_t(x) = \sum_{i=1}^{q} \big(\chi(i - 1 \le x < i) \cdot (t_i - t_{i-1})\big) + \chi(q \le x < q + 1) \cdot (t - t_q).
\end{equation*}
\end{Proposition}

\begin{proof}
If there is a group of nuts in each of the intervals $[0, 1)$, $[1, 2)$, ... $[k' - 1, k')$, then the searcher finds all nuts. This has a probability $\frac{k'!}{\lambda^{k'}} \ge \frac{k!}{\lambda^k}$.
\end{proof}

\begin{Theorem}\label{doublelimitk-1}
If $k \le 3$, then $v(k, k-1) = k!$~.
\end{Theorem}

\begin{proof}
The same strategy of the searcher as in $DLG(k, k)$ provides a lower bound of $k \cdot (k-1)! = k!$~.

If the hider chooses $(y_1, y_2, ..., y_k)$ uniformly randomly from the simplex $y_i \ge 0$, $\sum y_i = 1$, then the joint distribution of the $k - 1$ variables, say $(y_1, y_2, ..., y_{k - 1})$ is just a uniform random vector from the simplex $y_i \ge 0$, $\sum\limits_{i = 1}^{k-1} y_i \le 1$. Therefore, this shows an upper bound of $k \cdot (k-1)! = k!$~, as well.
\end{proof}

\subsection{The discrete limit game}

As we will see, the solution for the double limit game for $k = j$ is conjectured and partially proved to work if $h$ is larger than a constant. Therefore, it will be useful to define a limit game when $n \rightarrow \infty$ and $k$, $j$ and $h$ are constant. 

\begin{Def}
The \textbf{discrete limit game} $DG(k, j, h)$ is defined as follows. The hider chooses $k$ values $y_1, y_2, ... y_k \in [0, 1]$ where $\sum y_i \le 1$. Then for some subset $Q = \{ q_1, q_2, ..., q_j \} \subset \{1, 2, ..., k\}$ with a vector of different positive integers $x_{q_1}, x_{q_2}, ... x_{q_j} \in \Z^+$, the nuts are placed at $(x_{q_i}, y_{q_i})$. The searcher observes nothing. Independently from this, the searcher defines a strategy of the limit game with $\lambda = \infty$. The score of the searcher is the number of the vectors $x_{q_1}, x_{q_2}, ... x_{q_j}$ for which he wins, summing up for all different $Q$. This is what the searcher aims to maximize and the hider aims to minimize, in expectation. The value of this game is denoted by $v^*(k, j, h)$.
\end{Def}

\begin{Theorem}
\begin{equation*}
v^*(k, j, h) = \lim_{n \rightarrow \infty} v(k, j, n, h).
\end{equation*}
\end{Theorem}

\begin{proof}[Sketch of proof]
The proof will be similar to the proof for LG in Theorem~\ref{limitthm} with the following step from Theorem~\ref{2limitthm}.

In $G(k, j, n, h)$, by hiding the nuts into random holes with depths $\frac{1}{k}$, we can get that $v(k, j, n, h) = O(n^{-j})$. On the other hand, the searcher can get a score $\Omega(n^{-j+1})$ against placements that use the same hole for at least two nuts. E.g.\ $f_1(x) = \chi (x < 1)$ makes the job. Therefore, in any optimal hiding strategy, the probability that all nuts are placed in different holes should tend to 1.

Similar argument shows that if $n \rightarrow \infty$ and $\eps \rightarrow 0$, then the probability that the hider chooses a depth less than $\eps$ should also tend to 0. This implies that forbidding the searcher to dig in more than $\frac{1}{\eps}$ holes has a negligible effect on the value of the game.

Now we can convert the optimal hiding and searching strategy of $DG$ to strategies of $G$ with almost the same minimax and maximin scores with an error tending to 0. This works in the same way as in Theorem~\ref{limitthm}.
\end{proof}

\begin{Theorem} \label{discretelimitthm}
$v(k, j, n, h) \le {n + j - 1 \choose j} \cdot v^*(k, j, h)$
\end{Theorem}

\begin{proof}[Sketch of proof]
First, we bound the winning probability by the expected number of $j$-element subsets of nuts that would have been found by the searcher if the other nuts had not been cached.

Consider a hiding strategy of $DG$, and choose the same hiding strategy in G in the following sense. The hider chooses a uniformly random $k$-element multiset of holes out of the ${n + k - 1 \choose k}$ possibilities. We put the nuts into these holes in different depths as follows. We will have $k$ distances: the depth of the first nut in each non-empty hole, and the additional depths of the further nuts from the previous nuts. These depths will be a uniform random permutation of the random depths used by the optimal hiding strategy in $DG$.

Now any strategy of the searcher in $DG$ can be transformed to a strategy in $G$ by instead of digging a hole after finding a nut, the searcher digs in a new hole. This transformed strategy wins in the same number of cases.
\end{proof}

\begin{Conjecture} \label{discreteopt2}
If $k = j$ and for any $h$, if $n$ is large enough, then the bound in Theorem~\ref{discretelimitthm} is sharp, and hereby the transformed hiding strategy in G is optimal.
\end{Conjecture}

\begin{Question}
Is Conjecture~\ref{discreteopt2} true for all $k = j > 2$?
\end{Question}

\subsection{Solutions for the original problem for $k = j = 2$} \label{Alpern-original-sol}

In this section, unless we say the opposite, we \textbf{always assume that} $\boldsymbol{k = j = 2}$, namely, the hider caches two nuts, and the searcher aims to find both of them.

The following definition will simplify further analysis. 

\begin{Def}
A \emph{basic strategy} denoted by a pair $(y_1, y_2)$ means the following. The searcher chooses two holes, maybe the same hole twice, uniformly randomly out of the ${n+1 \choose 2}$ choices. If he chooses two different holes $x_1 \neq x_2$, then he caches the two nuts to $(x_1, y_1)$ and $(x_2, y_2)$, or $(x_1, y_2)$ and $(x_2, y_1)$, with the same probabilities. If he chooses the same hole $x$, then he caches the nuts to $(x, y_1)$ and $(x, 1)$, or $(x, y_2)$ and $(x, 1)$, randomly.\footnote{The strategy $(1, 0)$ can be replaced by the strategy of caching both nuts in the same random hole in depth 1.}
\end{Def}

\begin{Conjecture}
For any $n$ and $h$, there always exists an optimal hiding strategy which is a mixture of basic strategies.
\end{Conjecture}


In the light of this conjecture, we can use the solution of the double limit game (Theorem~\ref{limitsolution}) for the original game as follows.

\begin{Theorem} \label{discretedoublelimit}
If the hider uses the strategy $(y_1, y_2)$ for a uniform random pair satisfying $y_1 \ge 0$, $y_2 \ge 0$, $y_1 + y_2 \le 1$, then the searcher wins with probability at most $\frac{2h^2}{n(n+1)}$.
If $h \in \Z^+$ and $h \le \frac{n + 1}{2}$, then the bound is sharp, namely, the value of the game is $\frac{2h^2}{n(n+1)}$.
\end{Theorem}

The proof will be very similar to the proof of Theorem~\ref{limitsolution}. It will follow from Theorem~\ref{hiderstrategyk} and Theorem~\ref{searcherstrategy2}, about the strategies of the hider and the searcher.

\begin{Conjecture} \label{discretedoublelimitconj}
The bound $\frac{2h^2}{n(n+1)}$ in Theorem~\ref{discretedoublelimit} is sharp if $\frac{h^2}{\lfloor h \rfloor} \le \frac{n+1}{2}$ and either $h \ge 3$ or $h = 3 - \frac{1}{q}$ for any $q \in \Z^+ \setminus \{3\}$.
\end{Conjecture}

\begin{Theorem} \label{largeh}
If $\frac{n+1}{2} \le h$, then the value of the game is $\frac{\lfloor h \rfloor}{n}$. This is always an upper bound for the value of the game, because of the hiding strategy of putting both nuts at the same random hole, at depth 1.
\end{Theorem}

\begin{proof}
Hiding both nuts at the same hole in depth 1 is provides hiding probability at most $\frac{\lfloor h \rfloor}{n}$.

Consider now the following strategy of the searcher. He chooses $\lfloor h \rfloor$ holes at random, and starts digging in them parallelly, until a nut is found, at hole $x$ in depth $y$. Then he continues digging $x$ until depth $1$. Then if $y \le \frac{1}{2}$, then he digs the other $\lfloor h \rfloor - 1$ chosen holes until depth $1 - y$, and the remaining $n - \lfloor h \rfloor$ holes until depth $\min(y, 1-y)$.

If the nut with the higher depth (if the depths are the same, then either nut) is in one of the $\lfloor h \rfloor$ chosen holes, then the searcher finds both nuts. This has a probability at least $\frac{\lfloor h \rfloor}{n}$.

This strategy uses a total digging amount of
\begin{equation*}
1 + \big(\lfloor h \rfloor - 1\big) \cdot \max(y, 1-y) + \big(n - \lfloor h \rfloor\big) \cdot \min(y, 1-y)
\end{equation*}
\begin{equation*}
= 1 + (n - 1) \cdot \min(y, 1-y) + \big(\lfloor h \rfloor - 1\big) \cdot \big( \max(y, 1-y) - \min(y, 1-y) \big)
\end{equation*}
\begin{equation*}
\le 1 + (2h - 2) \cdot \min(y, 1-y) + (h - 1) \cdot \big( \max(y, 1-y) - \min(y, 1-y) \big)
\end{equation*}
\begin{equation*}
= 1 + (h - 1) \cdot \big( \max(y, 1-y) + \min(y, 1-y) \big) = 1 + (h - 1) = h. \qedhere
\end{equation*}
\end{proof}

\begin{Conjecture} \label{largehconj}
If $\frac{n+1}{2} \le \frac{h^2}{\lfloor h \rfloor}$, then the value of the game is $\frac{\lfloor h \rfloor}{n}$.
\end{Conjecture}

To challenge Conjectures \ref{discretedoublelimitconj} and \ref{largehconj}, or to try to prove them, the author suggests considering the following question.

\begin{Question}
For $n = 6$, $h = \sqrt{10.5} \approx 3.24$, is it true that the searcher can win with probability $\frac{1}{2}$?
\end{Question}

If $h < 3$, then the following discrete version of the searcher's double limit game solution can provide a better upper bound.

\begin{Theorem}\label{discretelimit}
If $h < \frac{a}{b}$ for some $a, b \in \Z^+$, then with the following hiding strategy, the searcher always wins with probability at most $\frac{2(a-1)(a-2)}{b(b-1) \cdot n(n+1)}$.

The hider chooses $y_1, y_2 \in \big\{\frac{1}{b}, \frac{2}{b}, \frac{3}{b}, ..., \frac{b - 1}{b}\big\}$, $y_1 + y_2 \le 1$ uniformly at random from the ${b \choose 2}$ possible choices, and chooses the hiding straregy $(y_1, y_2)$.
\end{Theorem}

\begin{proof}
The searcher can dig at most $a - 1$ possible hiding points (depths $\frac{1}{b}, \frac{2}{b}, ..., 1$). Given the strategy of searcher, if he finds the two nuts at the $i$th and $j$th searched possible hiding points, then it determines the two positions of the nuts. There are ${a - 1 \choose 2}$ different pairs of integers $1 \le i < j \le a - 1$, and there are ${b \choose 2} \cdot {n+1 \choose 2}$ possible pairs of positions, so the searcher cannot win with higher probability than $\frac{{a - 1 \choose 2}}{{b \choose 2} {n+1 \choose 2}} = \frac{2(a-1)(a-2)}{b(b-1) \cdot n(n+1)}$.
\end{proof}

\begin{Conjecture}
If $h \in \big(\frac{5}{2}, \frac{8}{3}\big) \cup \big[\frac{19}{7}, 2\big) \setminus \big\{3 - \frac{1}{q}\: q \in \Z^+\big\}$, then the best upper bound provided by Theorem~\ref{discretelimit} is sharp.
\end{Conjecture}

Now we have a conjecture of the solution for $h \in \big[\frac{5}{2}, n\big] \setminus \big[\frac{8}{3}, \frac{19}{7}\big)$.

For $h \in \big[0, \frac{9}{5}\big) \cup \big[2, \frac{7}{3}\big)$, the values of the games are the very same as for $n = 4$, written in the form $\frac{\alpha(h)}{n(n+1)}$. The proofs are also essentially the same.

\begin{Theorem} \label{solution2minus}
For $h \in \big[2 - \frac{1}{q-1},\ 2 - \frac{1}{q}\big)$, $q \in \{5, 6, 7, 8, 9\}$ and $n \le q - 1$, and for $h \in \big[\frac{9}{5}, 2\big)$ and $n \le 5$, the values of the games are $\frac{9}{2}, 5, \frac{26}{5}, \frac{28}{5}, \frac{17}{3}$, $6$, respectively, divided by $n(n+1)$.

An optimal hiding strategy in the first 5 cases are $\big(\frac{1}{q}, \frac{q-1}{q}\big)$ with probability $\frac{1}{2}, \frac{1}{2}, \frac{2}{5}, \frac{2}{5}, \frac{1}{3}$, respectively, and $\big(\frac{q-1}{2q}, \frac{q+1}{2q}\big)$ otherwise.  In the last case, it is $\big(\frac{1}{4}, \frac{3}{4}\big)$ with probability $\frac{2}{3}$ and $\big(\frac{1}{2}, \frac{1}{2}\big)$ with probability $\frac{1}{3}$, or in other words, it is a uniform random extremal strategy with depths multiples of $\frac{1}{4}$. An optimal strategy of the searcher is the mixture of the followings, until finding the first nut (the continuation is obvious, see Lemma~\ref{round2opt}). He caches a random hole until depth $h - 1$, then another one until depth $\frac{3 - h}{2}$, then continues the first hole until depth 1. Or he just caches a random hole until depth $1$.
The former strategy is used with probabilities $\frac{1}{4}, \frac{2}{4}, \frac{3}{5}, \frac{4}{5}, \frac{5}{6}, 1$, respectively.
\end{Theorem}

The proof is a simple but long case analysis which we omit from this paper.

For $h \in \big[\frac{7}{3}, \frac{5}{2}\big)$, we expect a similar but more difficult structure of the solutions as for $h \in \big[\frac{9}{5}, 2\big)$. What we know is the following.

\begin{Lemma} \label{val5/2}
If $h < \frac{5}{2}$, then the value of the game is at most $\frac{11}{n(n+1)}$. This can be achieved by the strategy $\big(\frac{1}{4}, \frac{3}{4}\big)$ with probability $\frac{1}{2}$ and $\big(\frac{1}{2}, \frac{1}{2}\big)$ with probability $\frac{1}{2}$.
\end{Lemma}

\begin{Conjecture} \label{conj5/2}
The bound in Lemma~\ref{val5/2} is optimal for $h \in \big[\frac{17}{7}, \frac{5}{2}\big)$.
\end{Conjecture}

\begin{Note} \label{note}
If one wants to solve $h \in \big[ \frac{7}{3}, \frac{17}{7} \big)$, then we suggest considering mixtures of extremal hiding strategies $\big(\frac{10 - 4h}{3}, \frac{4h - 7}{3}\big)$, $\big(\frac{h - 1}{5}, \frac{6 - h}{5}\big)$, $\big(\frac{16 - 6h}{5}, \frac{6h - 11}{5}\big)$, $\big(\frac{h - 1}{3}, \frac{4 - h}{3}\big)$.
\end{Note}

\begin{Theorem} \label{19/7}
If $h < h^* = \frac{67}{25}$ or $\frac{51}{19}$ or $\frac{19}{7}$, and $n \le 11$, then the searcher can win with probability at most $\frac{14\frac{2}{53}}{n(n+1)}$, $\frac{14\frac{2}{27}}{n(n+1)}$, $\frac{14\frac{2}{11}}{n(n+1)}$, respectively, if the hider uses the following mixture of hiding strategies.
\begin{itemize}
  \item $\big(\frac{3 - h^*}{2}, \frac{h^* - 1}{2}\big)$ with probability $\frac{12}{53}$, $\frac{20}{81}$, $\frac{4}{33}$, respectively;
  \item $\big(\frac{h^* - 1}{6}, \frac{7 - h^*}{6}\big)$ with probability $\frac{4}{53}$, $\frac{4}{81}$, $\frac{4}{33}$, respectively;
  \item $\big(\frac{h^* - 1}{4}, \frac{5 - h^*}{4}\big)$ with probability $\frac{36}{53}$, $\frac{56}{81}$, $\frac{8}{11}$, respectively;
  \item $\big(\frac{h^* - 1}{6}, \frac{h^* - 1}{6}\big)$ with probability $\frac{1}{53}$, $\frac{1}{81}$, $\frac{1}{33}$, respectively.
\end{itemize}

In particular, in the third case, the four depths are $\big(\frac{1}{7}, \frac{6}{7}\big)$, $\big(\frac{2}{7}, \frac{5}{7}\big)$, $\big(\frac{3}{7}, \frac{4}{7}\big)$ and $\big(\frac{2}{7}, \frac{2}{7}\big)$.
\end{Theorem}

The proof again is a simple but long case analysis, which we omit from this paper.

\begin{Conjecture} \label{19/7conj}
If $h \in \big[\frac{8}{3}, \frac{19}{7}\big)$, then the best bound in Theorem~\ref{19/7} is sharp.
\end{Conjecture}

The table summarizes our results and conjectures for $k = j = 2$.

\begin{figure*}[htb]
\footnotesize{
\begin{center}
\begin{tabular}[center]{|c||c|c|c|c|}
\hline
$h$ & $v(2, 2, n, h)$ & validity & status & notes \\
\hline \hline
$[0, 1)$ & $0$ & every $n$ & \multirow{17}{*}{proved} & \multirow{6}{*}{proved in earlier papers for $n = 4$,} \phantom{\Big(}\\
\cline{1-3}
$\big[1, \frac{3}{2}\big)$ & $\frac{2}{n(n+1)}$ & \multirow{2}{*}{$n \ge 2$} & \multirow{69}{*}{bound} & \multirow{5}{*}{the same proof works for $n \ge 4$} \phantom{\Big(}\\
\cline{1-2}
$\big[\frac{3}{2}, \frac{5}{3}\big)$ & $\frac{3}{n(n+1)}$ & & & \phantom{\Big(}\\
\cline{1-3}
$\big[\frac{5}{3}, \frac{7}{4}\big)$ & $\frac{4}{n(n+1)}$ & $n \ge 3$ & & \phantom{\Big(}\\
\cline{1-3}
$\big[\frac{7}{4}, \frac{9}{5}\big)$ & $\frac{4.5}{n(n+1)}$ & $n \ge 4$ & & \phantom{\Big(}\\
\cline{1-3}\cline{5-5}
$\big[\frac{9}{5}, \frac{11}{6}\big)$ & $\frac{5}{n(n+1)}$ & $n \ge 5$ & & \multirow{7}{*}{proved in Theorem~\ref{solution2minus}} \phantom{\Big(}\\
\cline{1-3}
$\big[\frac{11}{6}, \frac{13}{7}\big)$ & $\frac{5.2}{n(n+1)}$ & $n \ge 6$ & & \phantom{\Big(}\\
\cline{1-3}
$\big[\frac{13}{7}, \frac{15}{8}\big)$ & $\frac{5.6}{n(n+1)}$ & $n \ge 7$ & & \phantom{\Big(}\\
\cline{1-3}
$\big[\frac{15}{8}, \frac{17}{9}\big)$ & $\frac{5\frac{2}{3}}{n(n+1)}$ & $n \ge 8$ & & \phantom{\Big(}\\
\cline{1-3}
$\big[\frac{17}{9}, 2\big)$ & $\frac{6}{n(n+1)}$ & $n \ge 5$ & & \phantom{\Big(}\\
\cline{1-3}\cline{5-5}
$\big[2, \frac{11}{5}\big)$ & $\frac{8}{n(n+1)}$ & \multirow{2}{*}{$n \ge 4$} & & \multirow{2}{*}{proved in earlier papers for $n = 4$,} \phantom{\Big(}\\
\cline{1-2}
$\big[\frac{11}{5}, \frac{7}{3}\big)$ & $\frac{9}{n(n+1)}$ & & &  the same proof works for $n \ge 4$ \phantom{\Big(}\\
\hline
$\big[\frac{7}{3}, \frac{17}{7} \big)$ & ? & \multirow{3}{*}{$n \ge 5$} & open & see Note~\ref{note} \phantom{\Big(}\\
\cline{1-2}\cline{4-4}\cdashline{5-5}
$\big[\frac{17}{7}, \frac{5}{2}\big)$ & $\frac{11}{n(n+1)}$ & & \multirow{7}{*}{upper} & see Lemma~\ref{val5/2} and Conjecture~\ref{conj5/2} \phantom{\Big(}\\
\cline{1-3}\cline{5-5}
$\frac{5}{2}$ & $\frac{12.5}{n(n+1)}$ & \multirow{3}{*}{$n \ge 6$} & \multirow{9}{*}{bound} & see Theorem~\ref{discretedoublelimit} and Conj.~\ref{discretedoublelimitconj} \phantom{\bigg(}\\
\cline{1-2}\cline{5-5}
$\big(\frac{5}{2}, \frac{8}{3}\big)$ & $\inf\limits_{\frac{a}{b} > h}\frac{2(a-1)(a-2)}{b(b-1) \cdot n(n+1)}$ & & \multirow{11}{*}{proved} & see Theorem~\ref{discretelimit} \phantom{\bigg(}\\
\cline{1-3}\cline{5-5}
$\big[\frac{8}{3}, \frac{67}{25}\big)$ & $\frac{14\frac{2}{53}}{n(n+1)}$ & & & \multirow{5}{*}{see Theorem~\ref{19/7} and Conj.~\ref{19/7conj}} \phantom{\bigg(}\\
\cline{1-2}
$\big[\frac{67}{25}, \frac{51}{19}\big)$ & $\frac{14\frac{2}{27}}{n(n+1)}$ & $n \ge 11$ & &  \phantom{\bigg(}\\
\cline{1-2}
$\big[\frac{51}{19}, \frac{19}{7}\big)$ & $\frac{14\frac{2}{11}}{n(n+1)}$ & & & \phantom{\bigg(}\\
\cline{1-3}\cline{5-5}
$\big[\frac{19}{7}, 3\big)$ & $\inf\limits_{\frac{a}{b} > h}\frac{2(a-1)(a-2)}{b(b-1) \cdot n(n+1)}$ & $n \ge 8$ & & see Theorem~\ref{discretelimit} (and Conj.~\ref{discretedoublelimitconj}) \phantom{\bigg(}\\
\hline
$\big\{3, 4, ..., \big\lfloor \frac{n+1}{2} \big\rfloor\big\}$ & $ \frac{2h^2}{n(n+1)}$ & every $n$ & proved & proved in Theorem~\ref{discretedoublelimit} \phantom{\Big(}\\
\hline
$\big(3, \frac{n}{2} + O(1)\big]$ & $ \frac{2h^2}{n(n+1)}$ & $n \ge 6$ & \multirow{1}{*}{upper} & see Conjecture~\ref{discretedoublelimitconj} \phantom{\Big(}\\ 
\cline{1-3}\cline{5-5}
$\frac{n+1}{2} \le \frac{h^2}{\lfloor h \rfloor}$ & \multirow{3}{*}{$\frac{\lfloor h \rfloor}{n}$} & \multirow{3}{*}{every $n$} &  \multirow{1}{*}{proved} & see Conjecture~\ref{largehconj} \phantom{\Big(}\\
\cline{1-1}\cline{4-5}
$\big[\frac{n+1}{2}, n\big]$ & & & proved & proved in Theorem~\ref{largeh} \phantom{\Big(}\\
\hline
\end{tabular}
\end{center}
}
\end{figure*}

\subsection{Extensions for $k = j > 2$}

\begin{Theorem} \label{hiderstrategyk}
If $k \ge 2$, then
\begin{equation} \label{vkknh}
v(k, k, n, h) \le \frac{h^k}{{n + k - 1 \choose k}}.
\end{equation}
\end{Theorem}

\begin{proof}
We convert the proof for $v(k, k)$ (Theorem~\ref{limitsolution}) to a proof of this problem as follows. The hider chooses how many nuts to put to each hole, choosing one of the ${n + k - 1 \choose k}$ possibilities uniformly at random. Now we consider the distance of each nut from the closest nut above it, or if there is no nut above it, then the depth of the nut (the distance from the top). We choose these $k$ depths $y'_1, y'_2, ... y'_k$ uniformly at random from the simplex $y'_i \ge 0$ ($\forall i \in \{1, 2, ..., k'\}$), $\sum y'_i \le 1$. From here, we can continue with the proof of Proposition~\ref{hiderstrategy} with exchanging $\lambda^k$ to ${n + k - 1 \choose k}$.
\end{proof}

Now we show that \eqref{vkknh} is sharp if $k = 2$ and $h \in \N$.

\begin{Theorem} \label{searcherstrategy2}
If $k \in \N$
\begin{equation*}
v(2, 2, n, h) \ge \frac{h^k}{{n + k - 1 \choose k}}.
\end{equation*}
\end{Theorem}

\begin{proof}
Consider the following strategy of the searcher. He chooses $h$ holes at random, and he is digging them parallelly until a nut is found (but until at most depth 1, when the game ends). If a nut is found at a depth $y$, then he chooses $h$ new holes, as follows. With probability $\frac{2h}{n+1}$, he chooses the hole in which the nut was found, and the remaining $h - 1$ or $h$ holes are randomly chosen from the other $n - h$ holes. Then he digs these holes in depth $1 - y$ (more).

Assume first that the hider caches the two nuts in two different holes. If exactly one of them are in one of the $h$ holes the searcher started with, and the other one is in the next $h$ holes, then the searcher finds both nuts. Therefore, the searcher finds both nuts in expectedly $h \cdot \big(h - \frac{2h}{n + 1}\big)$ number of cases out of the $n \choose 2$ pairs, which happens with probability
\begin{equation*}
\frac{h \cdot \big(h - \frac{2h}{n + 1}\big)}{{n \choose 2}}
 = \frac{2h \cdot \frac{(n+1)h - 2h}{n+1}}{n(n-1)} = \frac{2h \cdot \frac{(n-1)h}{n+1}}{n(n-1)} = \frac{2h^2}{n(n+1)}.
\end{equation*}

Assume now that the hider caches the two nuts in the same hole. If this nut is in the first $h$ holes chosen by the searcher, and the searcher chooses to continue digging in this hole, then he finds both nuts. This has probability
\begin{equation*}
\frac{h}{n} \cdot \frac{2h}{n+1} = \frac{2h^2}{n(n+1)}.
\end{equation*}

To sum up, this strategy of the searcher finds both nuts with probability at least $\frac{2h^2}{n(n+1)}$, against any strategy of the hider.
\end{proof}

\begin{Conjecture}
If $k \ge 2$ and $\frac{k + 1}{k - 1} \le h \le \frac{n}{k}$, then $v(k, k, n, h) = \frac{h^k}{{n + k - 1 \choose k}}$.
\end{Conjecture}

Most probably, this conjecture can be proved for any integer $h$. So this weaker, seemingly easier conjecture is the following.

\begin{Conjecture}
If $k \ge 2$ and $h \in \Z^+$ and $h \le \frac{n}{k}$, then $v(k, k, n, h) = \frac{h^k}{{n + k - 1 \choose k}}$.
\end{Conjecture}

The author believes that the same proof works here as for $k = 2$, in Theorem~\ref{searcherstrategy2}. Except that when the searcher finds a nut and chooses $h$ new holes, then he might choose again any holes which had a nut at its current bottom. We should find the right probabilities for each of these choices so as for each distribution of the number of nuts in the different holes, the searcher finds them with the same probability.\footnote{Update: Dömötör Pálvölgyi recently proved this conjecture in \cite{palvolgyi2017all}.}

\medskip

Finally, we note that there are many other potentially useful limit problems not yet considered, when $k \rightarrow \infty$.
E.g. if $j_i \rightarrow \infty$ and $\frac{j_i}{k_i}$ is convergent. These limit games might also be very useful, and these may also look differently from the original game.

\section{Manickam--Miklós--Singhi Conjecture} \label{MMSsec}


As a much simpler application of the limit theory techniques, we show an easy way to generalize a well-known conjecture.

\begin{Problem}[MMS-Problem] \label{MMS-problem}
For a fixed $n, k \in \N$, find a sequence $a_i \in \R$, $a_1 + a_2 + ... + a_n = 0$ such that if $\{i_1, i_2, ..., i_k\} \subset \{1, 2, ..., n\}$ is a uniform random subset with cardinality $k$, then $\Pr\big( a_{i_1} + a_{i_2} + ... + a_{i_k} > 0 \big)$ is the largest possible.
\end{Problem}

Denote this maximum probability by $M(n, k)$. Two sequences are equivalent if after applying a permutation on one sequence, the two events that the sum is positive are the same.

\begin{Conjecture}[Manickam--Miklós--Singhi]
If $4k \le n$, then $M(n, k) = \frac{n-k}{n}$. The only optimal solution up to equivalence is $1-n, 1, 1, 1, ..., 1$.
\end{Conjecture}

The Manickam--Miklós--Singhi Conjecture was introduced in 1987 in \cite{manickam1987number}, and it has recently received a lot of attention, especially because of its connection to the Erdős matching conjecture\cite{erdHos1965problem}. In 1987, they proved the conjecture for $n > (k+1)(k^k + k^2) + k = \Theta(k^{k+1})$. In 2012, Tyomkin\cite{tyomkyn2012improved} proved it for $n > k^2 (4e \log k)^k = (\log k)^{\Theta(k)}$. In 2013, Huang and Sudakov\cite{huang2014minimum} proved it for $33k^2 < n$. In 2014, Chowdhury, Sarkis and Shahriari\cite{chowdhury2014manickam} proved it for $8k^2 < n$. Then in 2015, Pokrovskiy\cite{pokrovskiy2015linear} proved it for $k < \eps \cdot n$, but this improves the previous results only for $k > 10^{45}$.

\medskip

We consider a limit problem when $n \rightarrow \infty$, and $\frac{k}{n}$ converges. Finding the right limit problem is not an obvious task. One of the problems is that the solution $1-n, 1, 1, 1, ..., 1$ does not have a limit when $n \rightarrow \infty$. We resolve this problem by considering the following equivalent version of the MMS-problem.

\begin{Problem}[MMS-Problem-v2] \label{MMS-problem-v2}
For a fixed $n, k \in \N$, find a sequence $a_i \in \R$, such that if $\{i_1, i_2, ..., i_k\} \subset \{1, 2, ..., n\}$ is a uniform random subset with cardinality $k$, then $\Pr\Big( a_{i_1} + a_{i_2} + ... + a_{i_k} > \frac{k}{n}\sum\limits_{i=1}^{n} a_i \Big)$ is the largest possible.
\end{Problem}

This is clearly equivalent to the original problem, we only need to change each $a_i$ to $\frac{1}{n} \sum\limits_{i=1}^n a_i$. Now $-1, 0, 0, ..., 0$ is an equivalent form of the conjectured optimal solution, and we can say that the infinite sequence $-1, 0, 0, ...$ is a limit of it. We needed a few more observations to define the following limit problem.

\begin{Problem}[MMS-Limit] \label{limitproblem-simple}
For a fixed $p \in (0, 1)$, we are looking for a countable sequence $a_1, a_2, ...$ of real numbers with $\sum a_i^2 < \infty$ which maximizes $\Pr\big( \sum a_i (x_i - p) > 0 \big)$, where $x_1, x_2, ...$ are independent indicator variables with probability $p$.
\end{Problem}

There is a compact version of the limit problem, which is a kind of closure of the previous version. This will be a bit more difficult but easier to handle.

\begin{Problem}[MMS-CompactLimit] \label{limitproblem}
For a fixed $p \in (0, 1)$, we are looking for a countable sequence $a_1, a_2, ...$ of real numbers with $\sum a_i^2 < \infty$ and a real number $d$ which maximizes $\Pr\big( \sum a_i (x_i - p) + d x_0 > 0 \big)$, where $x_1, x_2, ...$ are indicator variables with probability $p$, and $x_0$ is a variable with standard normal distribution, and $x_0, x_1, x_2, ...$ are independent.
\end{Problem}

We call this problem a limit problem of the Manickam--Miklós--Singhi Conjecure, because the following theorem holds.

\begin{Theorem} \label{MMS-limit}
The supremum probabilities in MMS-Limit and MMS-CompactLimit are the same, denoted by $M(p)$, and for any $p \in (0, 1)$,
\begin{equation*}
\limsup_{\delta \rightarrow 0} M(p + \delta) = \limsup_{n_i \rightarrow \infty,\ \frac{k_i}{n_i} \rightarrow p} M(n_i, k_i).
\end{equation*}
\end{Theorem}

Before the proof, we show why this theorem is useful. We analyzed the limit problem for all values of $p$, the results are summarized in Figure~\ref{MMS-fig}. This leads to a new conjecture as follows.

\begin{figure}[]
{\scriptsize M(p)}
\begin{center}
\includegraphics[width=0.97\textwidth]{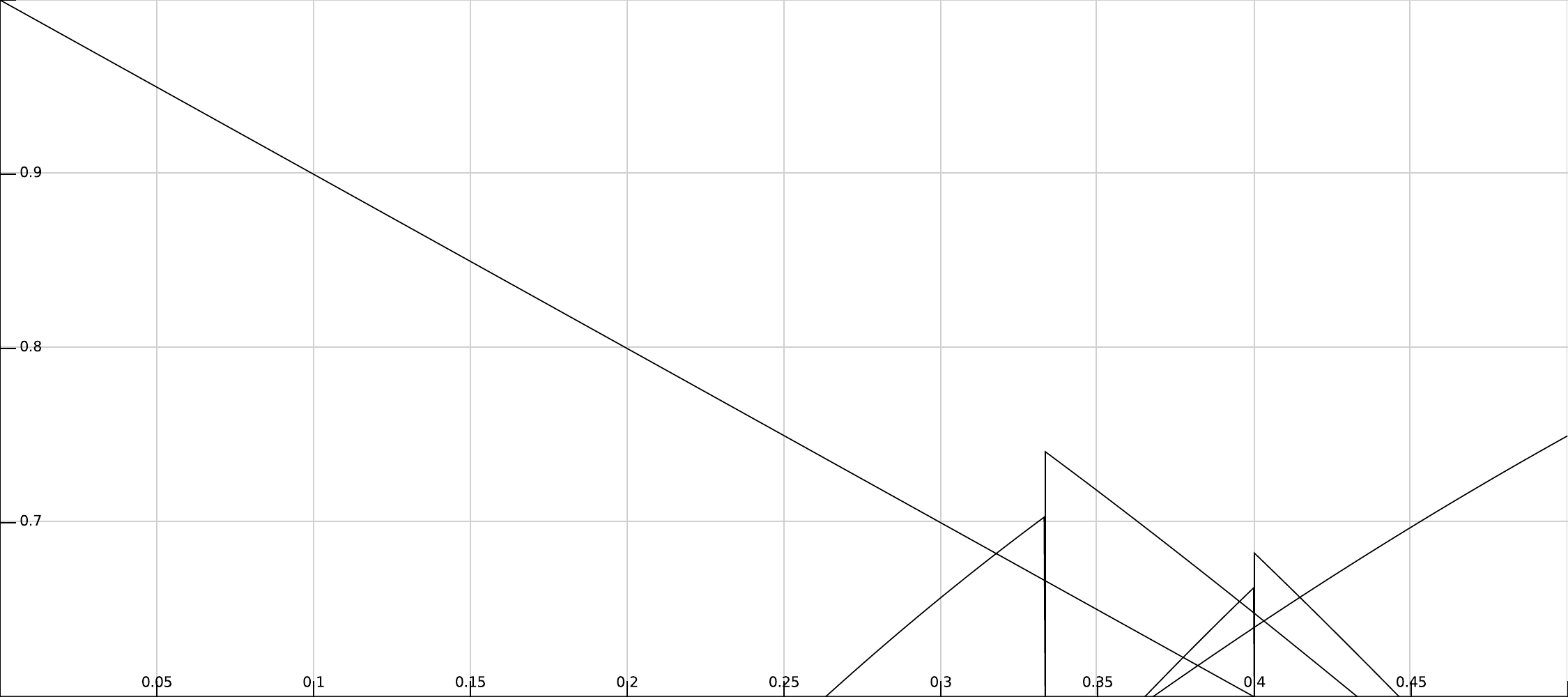}{\scriptsize \ p}
\end{center}
\caption{$M(p)$ is the maximum of the functions in the figure. If $p$ is between 0 and $0.317...$ (or exactly $1/3$), then the solution $-1$ is optimal (and the other coefficients are 0), between $0.317...$ and $1/3$ the optimal solution is $1, 1, 1$, between $1/3$ and $0.395...$ (and at $0.4$) the solution $-1, -1, -1$, between $0.395...$ and $0.4$ the solution $1, 1, 1, 1, 1$, between $0.4$ and $0.414...$ the solution $-1, -1, -1, -1, -1$, and between $0.414...$ and $0.5$ the solution $1, 1$ is optimal.}
\label{MMS-fig}
\end{figure}

\begin{Conjecture}
The optimal solution of the limit game has the form $a_1 = a_2 = ... = a_q$ where $q \in \{1, 2, 3, 5\}$, and all other coefficients are 0. This is the only optimal solution up to equivalence.
\end{Conjecture}

Now we are ready to form the corresponding conjecture for the original problem.

\begin{Conjecture}
The optimal solution of the original finite game has the form $a_1 = a_2 = ... = a_q$ where $q \in \{1, 2, 3, 5\}$, and $a_{q+1} = a_{q+2} = ... = a_n = - \frac{q}{n-q} a_1$. This is the only optimal solution up to equivalence.
\end{Conjecture}

The proof of Theorem~\ref{MMS-limit} will be based on the following version of the Central Limit Theorem, for which we omit the proof here.

\begin{Lemma} \label{CLT}
For every $p \in (0, 1)$ and $\eps > 0$, there exists $\delta > 0$ such that for any sequence $-\delta < a_1, a_2, ..., a_n < \delta$
and $\big| \frac{k}{n} - p \big| < \delta$ and $t \in \mathbb{R}$, the following holds.
If $\{i_1, i_2, ..., i_k\} \subset \{1, 2, ..., n\}$ is a uniform random subset with cardinality $k$, then
\begin{equation*}
\Phi\Big(\frac{t - \eps}{\sigma}\Big) - \eps < \Pr\Big( a_{i_1} + a_{i_2} + ... + a_{i_k} < \frac{k}{n} \sum_{i = 1}^n a_i + t \Big) < \Phi\Big(\frac{t + \eps}{\sigma}\Big) + \eps,
\end{equation*}
where $\sigma^2 = \operatorname{Var}\big(a_{i_1} + a_{i_2} + ... + a_{i_k}\big)$ and $\Phi$ is the distribution function of the standard normal distribution. \qedhere
\end{Lemma}

\begin{proof}[Proof of Theorem~\ref{MMS-limit}]
We can convert any solution of the MMS-Problem (Problem \ref{MMS-problem} or \ref{MMS-problem-v2}), MMS-Limit (Problem~\ref{limitproblem-simple}) and MMS-CompactLimit (Problem~\ref{limitproblem}) to each other. We only need to prove that the conversion error of the probabilities in question tends to 0 when $n \rightarrow \infty$.

\bigskip

A solution $\big((a_i), d\big)$ of MMS-CompactLimit can be converted to the solution of MMS-Limit by simply extending the sequence $(a_i)$ with $\big\lfloor\frac{d}{p(1-p)\eps}\big\rfloor$ number of $\eps$. The Central Limit Theorem shows that if $\eps \rightarrow 0$, then the sequence of the implied distributions converge to the distribution of $\sum a_i (x_i - p) + d x_0$. Therefore, the implied probabilities also tend to $\Pr\big( \sum a_i (x_i - p) + d x_0 > 0 \big)$ provided that $\Pr\big( \sum a_i (x_i - p) + d x_0 = 0 \big) = 0$. This holds if $d > 0$. If $d = 0$, then $(a_i)$ is already a perfect conversion.

\bigskip

Now we are converting a solution $a_1, a_2, ..., a_n$ of the MMS-problem-v2 to a solution of MMS-CompactLimit. We normalize the terms so that if we delete the largest and smallest $\eps n$ elements, then the average of the rest will be 0, and $\sum_{i = 1}^{n} a_i^2 = 1$. This does not change whether
$a_{i_1} + a_{i_2} + ... + a_{i_k} > \frac{k}{n}\sum a_i$.
We sort the elements $|a_1| \ge |a_2| \ge ... \ge |a_n|$. One of the following two cases must hold.

Case 1. There exists an index $l < 2 \log^2 n$ such that $|a_{l+1}| < 2 |a_{l + \log n}|$. In this case, with an approximation error tending to 0 as $n \rightarrow \infty$, the indices $1, 2, ..., l$ are chosen independently into $S = \{i_1, i_2, ..., i_k\} \subset \{1, 2, ..., n\}$, and Lemma~\ref{CLT} implies that $\sum\limits_{i \in S,\ i > l} a_i$ has an approximately normal distribution with approximately the same expected value $\frac{k}{n} \sum\limits_{i = l+1}^n a_i$ and variance $d$ for all $\big\{ i \le l, i \in S \big\}$. Therefore, $\big((a_1, a_2, ..., a_l), d \big)$ provides approximately the same probability in MMS-CompactLimit.

Case 2. For an $l > \log n$, there exists a sequence $j_1 < j_2 < ... < j_{2l} < 2 \log^2 n$ such that $2^i |a_{j_i}|$ is monotone decreasing. In this case, any $S \setminus \{j_1, j_2, ..., j_l\}$ can be extended in at most one way to get a set $S$ with $\Big| \sum\limits_{i \in S} a_i - \frac{k}{n} \sum\limits_{i=1}^{n} a_i \Big| < j_l$. Therefore, this event for a random $S$ has a probability tending to 0. $\sum\limits_{i = 2l+1}^{n} a_i < n \cdot j_{2l} < n \cdot 2^{-l} \cdot j_l < j_l$, therefore, if we exchange all elements $a_{2l+1}, a_{2l+2}, ..., a_n$ to 0, then this changes $\Pr\Big( \sum\limits_{i \in S} a_i > \frac{k}{n} \sum\limits_{i=1}^{n} a_i \Big)$ by asymptotically 0. The joint distribution of the events $a_i \in S$ for $i \le 2l$ are approximately the same as for independent events with probability $p$, therefore, $\Pr\Big( a_{i_1} + a_{i_2} + ... + a_{i_k} > \frac{k}{n} \sum\limits_{i=1}^{n} a_i \Big) \approx \Pr\Big( \sum\limits_{i=1}^{2l} a_i (x_i - p) > 0 \Big)$ with an error tending to 0.

\bigskip

A solution $(a_i)$ of MMS-Limit can be converted to MMS-problem as follows. Notice that a bounded subsequence of $(a_i)$ provides the same distribution with an arbitrary small error in $L_2$ norm. We either had a finite number of non-0 elements, or $\Pr\big( \sum a_i (x_i - p) = 0 \big) = 0$ (because we can choose an infinite sequence $j_1, j_2, j_3, ...$ such that $2^i |a_{j_i}|$ is monotone decreasing, and then we can use the technique in Case 2). Therefore, the change in $\Pr\big( \sum a_i (x_i - p) > 0 \big)$ tends to 0 if $n \rightarrow \infty$.
If we extend this subsequence with a large number of $0$ values to a total of $n$ elements, then choosing $\lfloor np \rfloor$ from them provides an approximately independent choice from the non-0 elements, with an arbitrary small error in the probabilities.
\end{proof}

\subsection{A suggestion for proving the MMS-conjecture}

Now we show a possible way to prove the conjecture by analyzing a very large but finite number of cases (most probably using a computer).
We say that a feasible normalized solution, or in short, a \emph{solution} for MMS-CompactLimit is a sequence and a number $\big((a_i)_{i \in \N}, d \big)$, where
\begin{equation*}
\operatorname{Var}\Big(\sum a_i (x_i - p) + d x_0 \Big) = p(1-p)\sum_{i \in \N} a_i^2 + d^2 = 1.
\end{equation*}

We define the distance between two solutions $\big((a_i), d_{\alpha}\big)$ and $\big((b_i), d_{\beta}\big)$ by
\begin{equation*}
  \dist \inf_{\sigma, \pi, k} \Big(\sum_{i=1}^{k} \big|a_{\sigma_i} - b_{\pi_i}\big| + \sup_{i>k}\big|a_{\sigma_i}\big| + \sup_{i>k}\big|b_{\pi_i}\big|\Big),
\end{equation*}
where $\sigma$ and $\pi$ are permutations on $\Z^+$.

\begin{Lemma}
The topological space of the solutions induced by $\dist$ is compact.
\end{Lemma}

\begin{Lemma}
If a solution $s = \big((a_i), d\big)$ satisfies $\Pr\big( \sum a_i (x_i - p) + d x_0 > -\eps \big) < v$, then there exists a neighborhood of $s$ in which for every solution $\big((a_i'), d'\big)$, we have $\Pr\big( \sum a_i' (x_i - p) + d' x_0 > 0 \big) < v$.
\end{Lemma}

\begin{proof}[Sketch of proof]
If ${\rm dist} \Big( \big((a_i), d_{\alpha}\big), \big((b_i), d_{\beta}\big) \Big)$ is small, then $\sum_{i=1}^{k} a_{\sigma_i} x_i \approx \sum_{i=1}^{k} b_{\pi_i} x_i$. About the rest of the terms, Lemma~\ref{CLT} shows that $\sum\limits_{i > k} a_{\pi_i} x_{\pi_i}$ has an approximately normal distribution, with approximately the same variance. Therefore, the two distributions are almost the same.
\end{proof}

With these two lemmas, we can hope that we can cover the solution space with a finite number of regions (open sets) so that we can show the conjectured inequality in each of these regions. Then we could modify these proofs so as to make it valid for the original discrete problem, as well.
For this, we would also need a version of Lemma~\ref{CLT} which gives an explicit $\delta > 0$ for each $\eps > 0$.

A similar but more detailed proof plan will be shown for the next related problem, in Section~\ref{KR-suggestion}.

\subsection{A new related question}

\begin{Problem} \label{problem-vectors}
There are $n, k \in \N$, and a vector space $V$ over $\R$ and $T \subset V$ is a convex closed set. We want to find vectors $a_1, a_2, ..., a_n$ with $a_1 + a_2 + ... + a_n = 0$ which maximizes $\Pr\big( a_{i_1} + a_{i_2} + ... + a_{i_k} \in T \big)$, where $\{i_1, i_2, ..., i_k\} \subset \{1, 2, ..., n\}$ is a uniform random $k$-element subset.
\end{Problem}

This problem can be reduced to the following problem.

\begin{Problem} \label{problem-numbers}
For fixed parameters $n, k \in \N$ and $1 < M \in \R$, find a sequence $a_i \in \R$, $a_1 + a_2 + ... + a_n = 0$ such that if $\{i_1, i_2, ..., i_k\} \subset \{1, 2, ..., n\}$ is a uniform random subset with cardinality $k$, then $\Pr\big( 1 \le a_{i_1} + a_{i_2} + ... + a_{i_k} \le M \big)$ is the largest possible.
\end{Problem}

The key lemma to show the equivalence of these problems is the following.

\begin{Lemma}
For each convex closed set $S \subset V$, $0 \notin S$, there exists two vectors $v, w \in V$ and positive real numbers $\lambda_1, \lambda_2 \in \R^+$ such that
$\lambda_1 v, \lambda_2 w \in S$ and $\forall x \in S\: x \cdot w \in [\lambda_1, \lambda_2]$.
\end{Lemma}

Problems \ref{problem-vectors} and \ref{problem-numbers} with $M = \lambda_2 / \lambda_1$ are equivalent because of the following reason. On one hand, for a sequence for Problem~\ref{problem-vectors}, the scalar product of the vectors with $w / \lambda_1$ provides a equally good sequence for Problem~\ref{problem-numbers}. On the other hand, a sequence for Problem~\ref{problem-numbers} multiplied by $\lambda_1 \cdot v$ provides an equally good solution for Problem~\ref{problem-vectors}.

Using the limit theory techniques in an analogous way, we can form the following conjecture.

\begin{Conjecture}
The optimal solution for Problem~\ref{problem-numbers} has the form $a_1 = a_2 = ... = a_q$ and $a_{q+1} = a_{q+2} = ... = a_n = - \frac{q}{n-q} a_1$. $q$ is bounded on $\frac{k}{n} \in [0, 1]\ \setminus [\frac{1}{2} - \eps,\ \frac{1}{2} + \eps]$ for all $\eps > 0$.
\end{Conjecture}

\section{Kikuta--Ruckle Conjecture} \label{KRsec}

We can use the same technique for the following generalization of the MMS-Problem, defined by Kikuta and Ruckle.\cite{kikuta2002continuous, alpern2013search}

\begin{Problem}[KR-Problem] \label{Kikuta}
$n, k \in \N$, and $d \in (0,1)$ are given. We want to find nonnegative real numbers $a_1, a_2, ..., a_n \ge 0$ with $a_1 + a_2 + ... + a_n = 1$ which maximizes $\Pr\big( a_{i_1} + a_{i_2} + ... + a_{i_k} > d \big)$, where $\{i_1, i_2, ..., i_k\} \subset \{1, 2, ..., n\}$ is a uniform random $k$-element subset.
\end{Problem}

Denote this supremum\footnote{We believe that this is a maximum. If we use ``$\ge d$'' rather than ``$> d$'', than due to the compactness of the space of solutions, we can even prove it.} probability by $K(k, n, d)$. Notice that if $d < \frac{k}{n}$, then $a_i = \frac{1}{n}$ provides $K(k, n, d) = 1$ and if $d = \frac{k}{n}$, then we get back the MMS-Problem.

Furthermore, we can get the Kikuta--Ruckle problem from Alpern's Caching Game by the following modification and by taking the limit $k \rightarrow \infty$ and $\frac{k}{n} \rightarrow d$. This modification is that we replace the overall hiding time limit to the restriction that the searcher cannot use a depth more than 1 (or other than 1), and we consider the limit $k \rightarrow \infty$ and $\frac{j}{k} \rightarrow d$.

\begin{Conjecture}[Kikuta--Ruckle]
For all $n, k \in \N$, and $d \in (0, 1)$, there is an optimal solution for the KR-Problem of the form $a_1 = a_2 = ... = a_s = \frac{1}{s}$ and $a_{s+1} = a_{s+2} = ... = a_n = 0$ for some $s \in \{1, 2, ..., n\}$.
\end{Conjecture}

The conjecture says nothing about the optimal value of $s$. The authors as well as other researchers on the topic found a very chaotic behaviour of this value. However, it would be useful to know the value if we want to prove the conjecture. Searching for the optimal values for small constant values $n, k$ did not really help, we will shortly see the reason of it. Instead, we will consider what happens if $n \rightarrow \infty$ and hereby we form a conjecture about the value of $s$.

The KR-Problem has one more parameter than the MMS-Problem, therefore, we have a larger freedom about defining limit problems of it with $n_i \rightarrow \infty$. The simplest limit problem is when $\frac{k_i}{n_i} \rightarrow p \in (0, 1)$ and $d_i = d$, which is simply the following.

\begin{Problem} \label{KR-simplelimit}
For fixed parameters $0 < p < d < 1$, we are looking for a countable sequence $0 \le a_1, a_2, ...$ of \emph{nonnegative} real numbers with $\sum a_i = 1$ which maximizes
$\Pr\big( \sum a_i x_i > d \big)$,
where $x_1, x_2, ...$ are independent indicator variables with probability $p$.
\end{Problem}

Denote this maximum probability by $K(p, d)$.

\begin{Theorem}
For any $p, d \in (0, 1)$,
sequence $(n_i, k_i, d_i)$, ,
\begin{equation*}
\limsup_{n_i \rightarrow \infty;\ \frac{k_i}{n_i} \rightarrow p;\ d_i \rightarrow d} K(n_i, k_i, d_i) = \limsup_{\eps, \delta \rightarrow 0} K(p + \eps, d + \delta).
\end{equation*}
\end{Theorem}

\begin{proof}[Sketch of proof]
As in the proof of Theorem~\ref{MMS-limit}, we convert the solutions of the KR-problem to solutions of the limit problem and vice versa, with conversion errors $\rightarrow 0$ as $n \rightarrow \infty$.
\end{proof}

The proof that it is a limit problem.
We only note that the key observation is that if $\max a_i$ is small, then $\rm{Var}\big(\sum a_i x_i\big)$ is also small, therefore, with a large probability, $\sum a_i x_i \approx p < d$ which cannot be optimal.
Based on the conjectured solution of Problem~\ref{KR-simplelimit}, the conjectured value of $s$ is described in Figure~\ref{bigpicture} (extended with the simple case $d < \frac{k}{n}$).

\begin{figure}[]
\begin{center}
\includegraphics[width=\textwidth, trim=50mm 20mm 50mm 20mm]{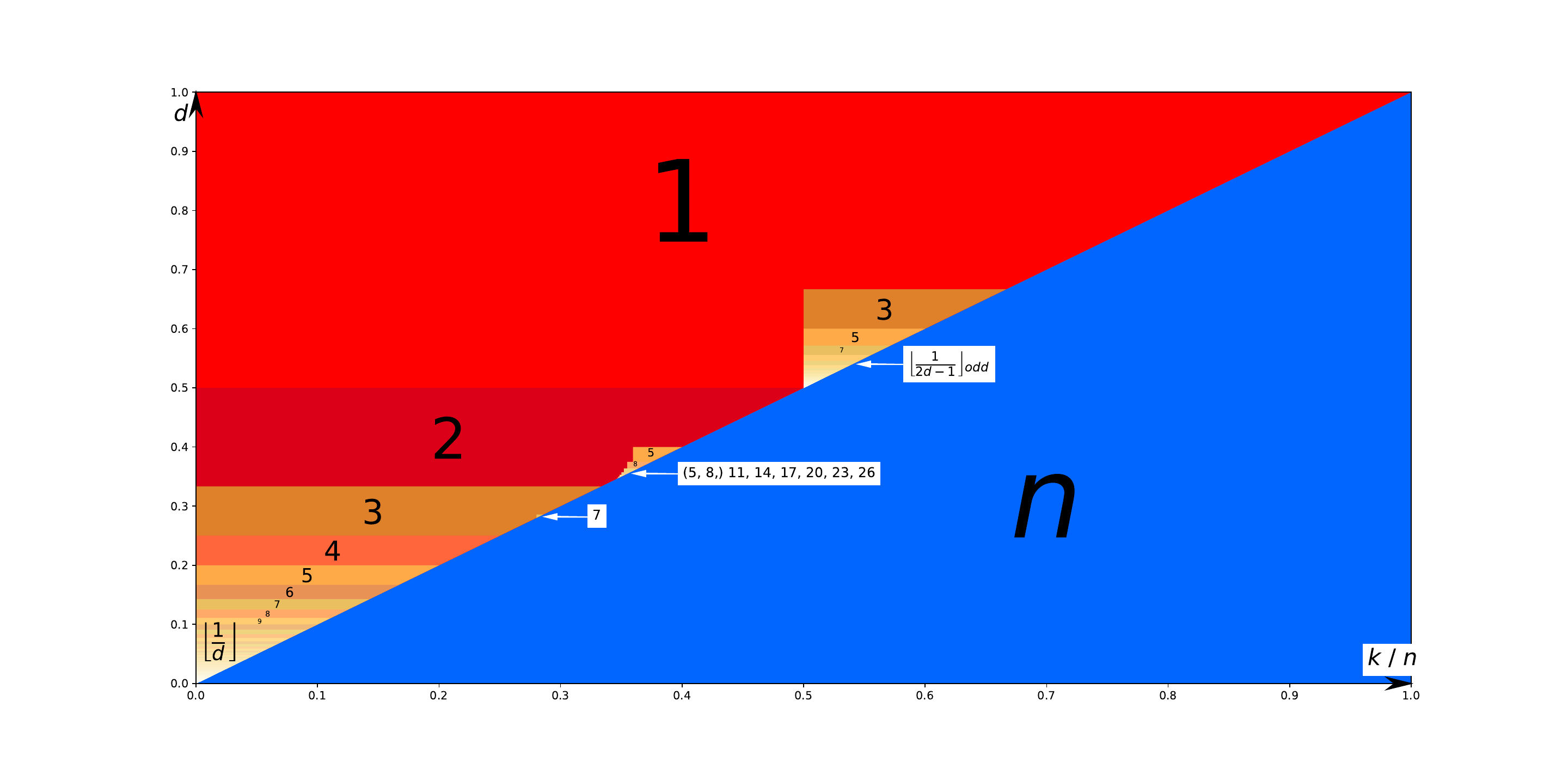}

\caption{The conjectured asymptotical solution for fixed $d$ and converging $\frac{k}{n}$.}
\label{bigpicture}
\end{center}
\end{figure}

However, it misses some very important cases: when $d$ is just above $\frac{k}{n}$.
Now we define some other limit problems about this case. When $n_i \rightarrow \infty$, $\frac{k_i}{n_i} \rightarrow p \in (0, 1)$, $d_i = \frac{k_i}{n_i}$, then we just get to MMS-CompactLimit. Let us see what happens if we increase $d$. If $d > \frac{k}{n-1}$, then it is no longer useful to choose $s = n - 1$ (which corresponds to the solution $-1, 0, 0, ...$ for the MMS-problem). Moreover, if $d > \frac{k}{n-c}$ for any constant $c$, then $s \ge n - c$ provides $\Pr( a_{i_1} + a_{i_2} + ... + a_{i_k} > d ) = 0$. Therefore, if $n_i \rightarrow \infty$, $\frac{k_i}{n_i} \rightarrow p \in (0, 1)$, $d_i \rightarrow p$ and $n_i d_i - k_i \rightarrow \infty$, then this leads to the following limit problem.

\begin{Problem} \label{Kikuta-limit}
For a fixed $p \in (0, 1)$, we are looking for a countable sequence $0 \le a_1, a_2, ...$ of \emph{nonnegative} real numbers with $\sum a_i^2 < \infty$ and a real number $\sigma$ which maximizes
\begin{equation} \label{K-l-eq}
\Pr\big( \sum a_i (x_i - p) + \sigma x_0 > 0 \big),
\end{equation}
where $x_1, x_2, ...$ are indicator variables with probability $p$, and $x_0$ is a variable with standard normal distribution, and $x_0, x_1, x_2, ...$ are independent.
\end{Problem}

Denote this supremum probability by $K(p)$.
Figure~\ref{KR-compare} shows the efficiency of the conjectured optimal solutions.

\begin{figure}[]
{\scriptsize K(p)}
\begin{center}
\includegraphics[width=0.97\textwidth]{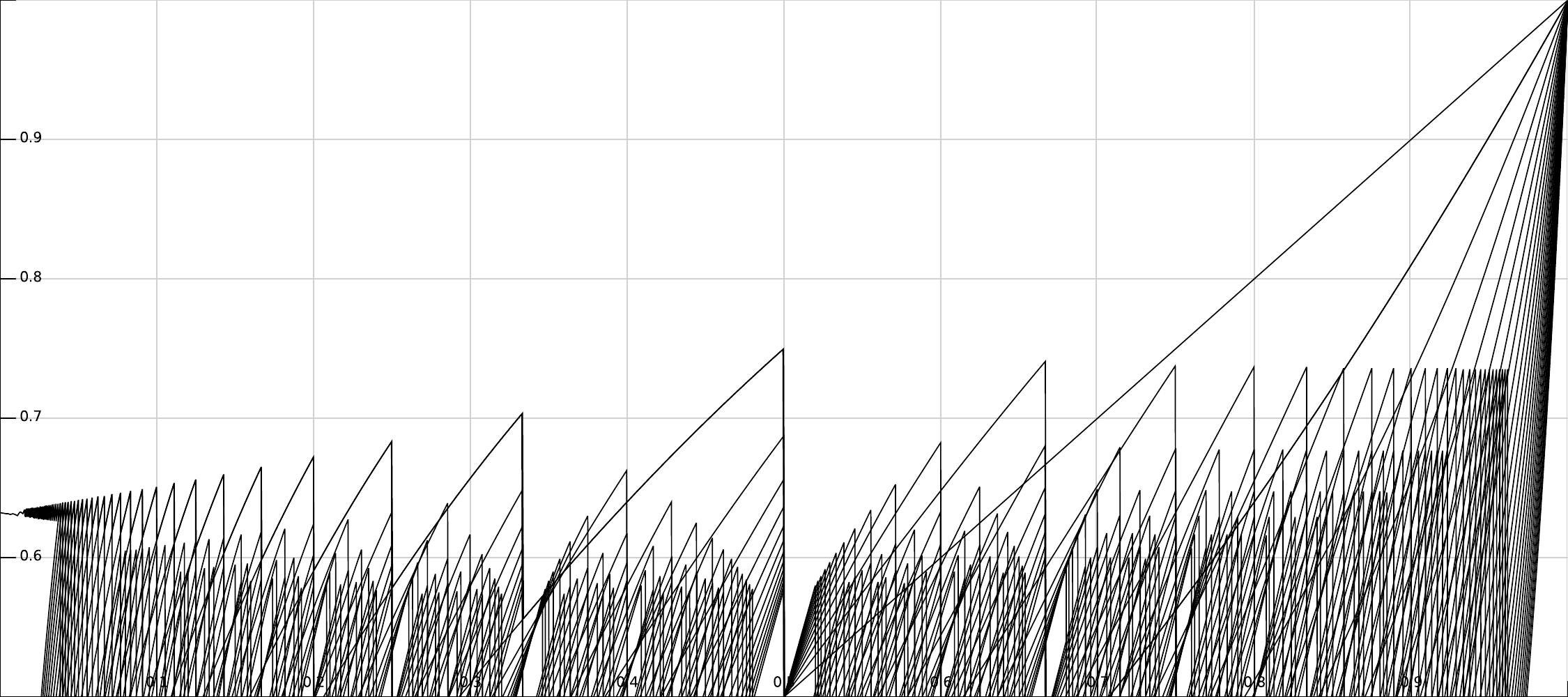}{\scriptsize \ p}

\caption{The efficiency of the solutions when $n \rightarrow \infty$, $\frac{k_i}{n_i} \rightarrow p$, $d_i \rightarrow p$ but $d_i \cdot n_i - k_i \rightarrow \infty$.}
\label{KR-compare}
\end{center}
\end{figure}

Again, the reason why we consider it a true limit problem is the following theorem.

\begin{Theorem}
For any sequence $(n_i, k_i, d_i)$, $n_i \rightarrow \infty$ and $\frac{k_i}{n_i} \rightarrow p \in (0, 1)$, $d_i \rightarrow p$ and $n_i d_i - k_i \rightarrow \infty$,
\begin{equation*}
\liminf_{\delta \rightarrow 0} K(p + \delta) \le \liminf K(n_i, k_i, d_i) \le
\limsup K(n_i, k_i, d_i) \le \limsup_{\delta \rightarrow 0} K(p + \delta).
\end{equation*}
\end{Theorem}

\begin{proof}[Sketch of proof] The proof is the same as for Theorem~\ref{MMS-limit} with the additional observation that the median is at most $\frac{2}{n}$, and changing a few terms $a_i$ by $O\big(\frac{1}{n}\big)$ has a negligible effect.
\end{proof}

The following limit problem describes the transition between MMS-Limit and \ref{Kikuta-limit}. If $0 < \frac{k}{n} < d \rightarrow 0$, with different values of $\lim d n - k = \lambda$, then it has the same limit as Problem~\ref{Kikuta-limit} except that \eqref{K-l-eq} is replaced to the following.
\begin{equation*}
\Pr\big( \sum a_i (x_i - p) + \sigma x_0 > -\lambda \min_i a_i \big)
\end{equation*}

Its conjectured solution is shown in Figure~\ref{above-diagonal-fig}.

\begin{figure}[] \label{above-diagonal}
\begin{center}
\includegraphics[width=\textwidth, trim=50mm 20mm 50mm 20mm]{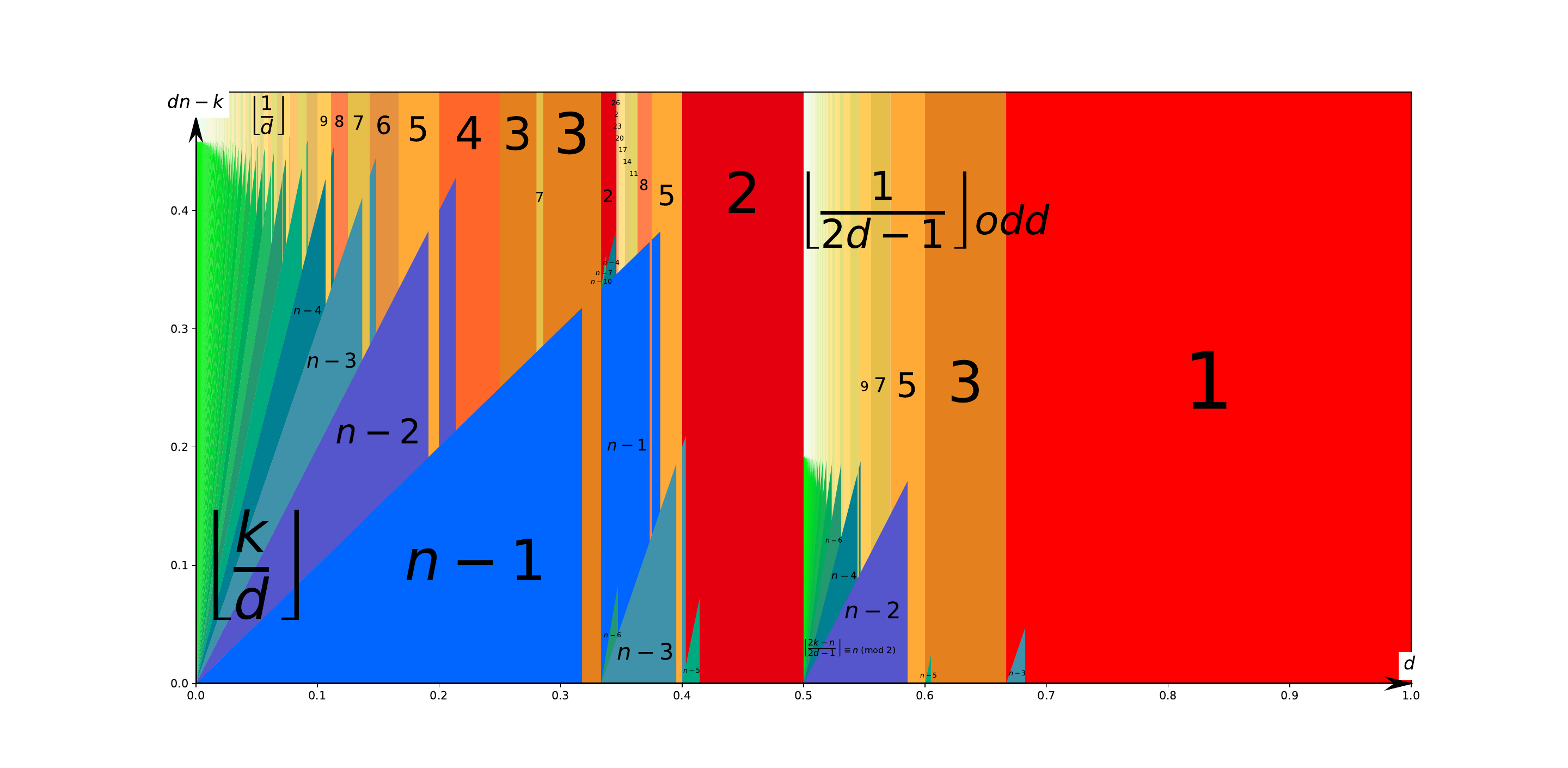}

\caption{The conjectured asymptotical solution for $n_i \rightarrow \infty$, $d_i \rightarrow d$ and convergent $d_i \cdot n_i - k_i$.}
\label{above-diagonal-fig}
\end{center}
\end{figure}


After analyzing all limits, we can make a conjecture about how $s$ depends on $n$ and $d$. It can be read from the figures, but here is a summary.

\begin{Conjecture}
For all $n, k \in \N$, and $d \in (0, 1)$, there is an optimal solution for the KR-Problem of the form $a_1 = a_2 = ... = a_s = \frac{1}{s}$ and $a_{s+1} = a_{s+2} = ... = a_n = 0$, where\footnote{$\lfloor x \rfloor = \max \{y \in \Z : y < x\}$ and $\lfloor x \rfloor_{\text{odd}} = \max \{ y \equiv 1 \ (\text{mod }2), y < x\}$ and $\lfloor x \rfloor_{\equiv n \ (\text{mod }2)} = \max \{ y \equiv n \ (\text{mod }2), y < x\}$. If we define the KR-Problem with "$\le$" instead of "$<$", then here we have $y \le x$.}

\begin{equation*}
s \in \Big\{ \Big\lfloor \frac{1}{d} \Big\rfloor, \Big\lfloor \frac{1}{2d - 1} \Big\rfloor_{\textnormal{odd}}, 5, 7, 8, 11, 14, 17, 20, 23, 26, \Big\lfloor \frac{k}{d} \Big\rfloor, \Big\lfloor \frac{2k-n}{2d-1} \Big\rfloor_{\equiv n \ (\textnormal{mod }2)},
\end{equation*}
\begin{equation*}
n - 10, n - 7, n - 6, n - 5, n - 4, n - 3, n\Big\}.
\end{equation*}
\end{Conjecture}

We show an interesting case about the final analysis of the KR-problem. When $\frac{10}{29} > d > \frac{k}{n} \rightarrow \frac{10}{29}$, then $s = 29$ provides probability tending to $0.5694$, but $s = 2$ provides a very slightly higher limit probability $0.5707$. However, when $n$ is small, then the solution $s = 29$ slightly beats $s = 2$, thanks to the negative correlations between the events of choosing the different terms. There were 167 different pairs $(n, k)$ of this kind, in all cases $\frac{k}{n}$ is very close to $\frac{10}{29}$. In these cases, $\frac{k}{n-1} > \frac{k-1}{n-4} > \frac{10}{29}$ holds, and therefore, $s = n - 1$ or $s = n - 4$ can beat some of these solutions. In 106 cases both beat $s = 29$, in 60 cases only $s = n - 1$ beats $s = 29$ and in only one case, when $n = 62$, $k = 21$, only $s = n - 4$ beats $s = 29$.
To sum up, the solution $s = 29$ was the only serious candidate for being optimal only in a finite number of cases (where $\frac{k}{n}, d \approx \frac{10}{29}$), but at least one of $s = 2$ or $s = n - 1$ or $s = n - 4$ always beats it.

The author found this technique very useful for seeking for counterexamples, as well. Now he strongly believes that the conjecture is true, but it is rather ``accidentally true" and he doubts that there exists a simple proof. He found that the best candidates for counterexamples use the terms $\frac{2}{s}$, $\frac{1}{s}$ and $0$ for some $s$. Showing that there is no counterexample of this form is already a very difficult task, we need completely different arguments for the different cases.

\subsection{Compact limits and a suggestion for proving the conjecture} \label{KR-suggestion}

First, we are trying to prove the conjectures for the limit games. The hope is that we can modify it to a proof of the original problem.
The suggested proof would use a large number of cases, probably analysed using a computer. We consider the following equivalent version of Problem~\ref{KR-simplelimit}, and we show the analysis of a few cases.

\begin{Problem}
For fixed parameters $p, d \in (0, 1)$, we are looking for a countable sequence $0 \le a_1, a_2, ...$ of \emph{nonnegative} real numbers with $\sum a_i \le 1$ which maximizes
\begin{equation*}
\Pr\bigg( \sum a_i x_i + p\Big(1 - \sum a_i \Big) > d \bigg),
\end{equation*}
where $x_1, x_2, ...$ are independent indicator variables with probability $p$.
\end{Problem}

This is the limit of the sequence $a_i$ extended by $\frac{1 - \sum a_i}{\eps}$ number of $\eps$ elements, when $\eps \rightarrow 0$. This is analogous to the MMS-Problem and MMS-Problem-v2 (Problems \ref{MMS-problem} and \ref{MMS-problem-v2}).

\emph{Solution} means an countable sequence $(a_i)$ with $a_1 \ge a_2 \ge ... \ge 0$ and $\sum a_i \le 1$.
We define the distance between two solutions $(a_i)$ and $(b_i)$ by
\begin{equation*}
  {\rm dist}\big((a_i), (b_i) \big) = \inf_{k \in \Z^+} \Big(\sum_{i=1}^{k-1} \big|a_i - b_i\big| + a_k + b_k\Big).
\end{equation*}

\begin{Theorem}
The topological space of the solutions induced by {\rm dist} is compact.
\end{Theorem}

\begin{Theorem}
For any fix parameter $p \in (0, 1)$, we assign the distribution $\sum a_i x_i + p\big(1 - \sum a_i \big)$, where $x_1, x_2, ...$ are independent indicator variables with probability $p$. This mapping is continuous with respect to the topology induced by {\rm dist} and the topology on the space of probability distributions in $\R$.
\end{Theorem}

We omit the proofs of these theorems, but the anticipated proof of the Kikuta--Ruckle Conjecture would not need this proof.

Let us consider the problem with $p = 0.2$. The conjectured maximum probability in Problem~\ref{KR-simplelimit} is $f(d) = 1 - 0.8^{\lfloor 1/d \rfloor}$ if $d \ge 0.2$, and $f(d) = 1$ if $d < 0.2$. This is a proved lower bound given from $s = \big\lfloor \frac{1}{d} \big\rfloor$ or $a_i \equiv 0$. Assume by contradiction that this bound is not tight. This means that, for some $\eps > 0$, the true value of the maximum probability $F(d)$ satisfies $F(d) \le f(d) + \eps$ with equation for at least one value of $d \in \big\{ \frac{1}{2}, \frac{1}{3}, \frac{1}{4}, \frac{1}{5} \big\}$. The plan is to prove that we do not have equation in any of the four cases.

We show a few cases of the proof for $d = \frac{1}{2}$, namely, $F\big(\frac{1}{2}\big) < 0.2 + \eps$.

If $a_1 > \frac{1}{2}$, then $\sum a_i x_i + 0.2 \cdot \big(1 - \sum a_i \big) > \frac{1}{2}$ implies $x_1 = 1$, which has probability $0.2$.

If $a_1$ is small enough, then $\sum a_i x_i + 0.2 \cdot \big(1 - \sum a_i \big)$ is concentrated around $0.2$, and therefore, it is more than $0.5$ with probability less than $0.2$.

As a nontrivial case, if $a_1 + a_2 \ge \frac{1}{2} \ge 2 a_2 - a_1$, then the following argument is valid.
We split the probability distribution of $(x_i)$ into the convex combination of the following four distributions, including their weights (or probabilities). The splits will be independent from the values $x_3, x_4, ...$ . Let $X$ denote the event $\sum a_i x_i + 0.2 \cdot \big(1 - \sum a_i \big) > \frac{1}{2}$.

\emph{Event 1 (with $4\%$).} $x_1 = x_2 = 1$. Then $\Pr(X) = 1$.

\emph{Event 2 (with $64\%$).} $x_1 = x_2 = 0$. Then $\Pr(X) = 0$.

\emph{Event 3 (with $12\%$).} $x_1 = 1$ and $x_2 = 0$. Then $\Pr(X) \le 1$.

\emph{Event 4 (with $20\%$).} With probability $0.8$, $x_1 = 0$ and $x_2 = 1$ and with probability $0.2$, $x_1 = 1$ and $x_2 = 0$. Now $X$ is equivalent to the event that
\begin{equation} \label{trans-a}
(a_1 - a_2) x_1 + \sum_{i \ge 3} a_i x_i + 0.2 \Big(1 - \sum a_i \Big) > \frac{1}{2} - a_2,
\end{equation}
where $x_1, x_3, x_4, x_5, ...$ are independent indicator variables with probability $0.2$. Therefore, if we define $b_1 = \frac{a_1 - a_2}{1 - 2 a_2}$, $b_2 = \frac{a_3}{1 - 2 a_2}$, $b_3 = \frac{a_4}{1 - 2 a_2}$, ..., then \eqref{trans-a} is equivalent to
$\sum b_i x_i + 0.2 \Big(1 - \sum b_i \Big) > \frac{1}{2}$, which has probability at most $0.2 + \eps$.

Therefore, $\Pr(X) \le 4\% \cdot 1 + 64\% \cdot 0 + 12\% \cdot 1 + 20\% \cdot (0.2 + \eps) = 0.2 + 0.2 \eps < 0.2 + \eps$.

\section{Concluding remarks}

When we want to solve a discrete problem, then one of the techniques we often try is to define and solve an infinite version of the problem, hoping that it helps with the original problem.
Here, an infinite version is defined by just exchanging some finite parameters to infinity and making the necessary modifications so as to get a meaningful problem.
In this paper, we showed a better way of doing this.
Namely, we are looking for a problem for which we can prove that its solution is the limit of the solutions of the finite problems.
This is what we call limit problem.

Our technique can be summarized as a 3-step plan, all these steps are generally not easy but easier than solving the original problem directly.
\begin{enumerate}
  \item Find interesting limit problems. This means that we need the convergence but we do not need similarity to the original problem. The simpler the limit problem the better. Having a compact solution space can be useful. (Complete proofs of the convergences are not necessary.)
  \item Solve the limit problems.
  \item Solve the original problem using the solutions of the limit problems.
\end{enumerate}

\section*{Acknowledgements}

I am thankful to Dezső Miklós for his observations and suggestions about the MMS-problem.

\bibliography{caching}

\end{document}